\documentclass[smallcondensed]{svjour3}[]
\usepackage[utf8]{inputenc}
\usepackage{graphicx}
\usepackage{amssymb}
\usepackage[cmex10]{amsmath}
\usepackage{color}
\usepackage{theorem}
\usepackage{pifont}
\usepackage{hyperref}
\hypersetup{colorlinks=false,pdfborderstyle={/S/U/W 1},pdfborder=0 0 1}
\usepackage{cite}
\usepackage{url}
\usepackage{psfrag}
\usepackage{bbm}
\usepackage{array}
\usepackage{booktabs}
\usepackage{nicefrac}
\usepackage[english]{babel}
\usepackage[english = american]{csquotes}
\usepackage{mathtools}
\MakeOuterQuote{"} 
\RequirePackage[ruled,vlined,titlenumbered,linesnumbered]{algorithm2e}

\newtheorem{definition_it}{Definition}
\newtheorem{problem_bold}{Problem}
\newtheorem{example_bold}{Example}

\let\myOmega\Omega
\renewcommand{\Omega}{\mathrm{\myOmega}}
\let\myLambda\Lambda
\renewcommand{\Lambda}{\mathrm{\myLambda}}
\let\myDelta\Delta
\renewcommand{\Delta}{\mathrm{\myDelta}}
\let\myPhi\Phi
\renewcommand{\Phi}{\mathrm{\myPhi}}
\let\myPsi\Psi
\renewcommand{\Psi}{\mathrm{\myPsi}}
\let\myGamma\Gamma
\renewcommand{\Gamma}{\mathrm{\myGamma}}

\newcommand{\Fqm}{\ensuremath{\mathbb F_{q^m}}}

\newcommand{\Fq}{\ensuremath{\mathbb F_{q}}}

\newcommand{\myset}[1]{\mathcal{#1}}
\newcommand{\intervallincl}[2]{\ensuremath{[#1,#2]}}
\newcommand{\intervallexcl}[2]{\ensuremath{[#1,#2-1]}}
\newcommand{\setelements}[2]{\ensuremath{\{#1_0,#1_1,\dots, #1_{#2-1}\}}}

\newcommand{\Basis}{\ensuremath{\myset{B}}}

\newcommand{\NormbasisOrdered}{\boldsymbol{\Normelement}}
\newcommand{\NormbasisDualOrdered}{\boldsymbol{\Normelement}^{\perp}}

\newcommand{\Normelement}{\beta}

\newcommand{\Normdualelement}{{\beta^{\perp}}}

\newcommand{\qreciproc}[1]{\overline{#1}}
\newcommand{\qtrafo}[1]{\widehat{#1}}
\newcommand{\Linpolyring}{\mathbb{L}_{q^m}[x]}

\newcommand{\mymap}[1]{\textup{{#1}}}

\newcommand{\OCompl}[1]{\ensuremath{\mathcal{O}({#1})}}

\newcommand{\printalgo}[1]
{\begin{center}
\vspace{-2ex}
\scalebox{0.87}{
\begin{tabular}{p{0.94\textwidth}}
\begin{algorithm}[H]
 #1
\end{algorithm}
\end{tabular}}
\vspace{-2ex}
\end{center}}

\newcommand{\printalgoWidth}[2]
{\begin{center}
\vspace{-2ex}
\scalebox{0.85}{
\begin{tabular}{p{#2\textwidth}}
\begin{algorithm}[H]
 #1
\end{algorithm}
\end{tabular}}
\vspace{-2ex}
\end{center}}

\DeclareMathOperator{\defi}{def}
\newcommand{\defeq}{\overset{\defi}{=}}

\renewcommand{\mod}{\; \textnormal{ mod } \;}

\newcommand{\Indfunc}{\mathbf{1}}

\DeclareMathOperator{\rk}{rk}

\renewcommand{\vec}[1]{\ensuremath{\mathbf{#1}}}
\newcommand{\Mat}[1]{\ensuremath{\mathbf{#1}}}
\newcommand{\vecelements}[1]{\ensuremath{(#1_0 \ #1_1 \ \dots \ #1_{n-1})}}
\newcommand{\vecelementsm}[1]{\ensuremath{(#1_0 \ #1_1 \ \dots \ #1_{m-1})}}

\newcommand{\veceval}[2]{\ensuremath{\left(#1(#2_0) \ #1(#2_1) \ \dots \ #1(#2_{n-1})\right)}}
\newcommand{\vecevalm}[2]{\ensuremath{\left(#1(#2_0) \ #1(#2_1) \ \dots \ #1(#2_{m-1})\right)}}

\newcommand{\Mooremat}[2]{\mymap{qvan}_{#1}( #2 )}

\newcommand{\MoormatExplicit}[3]{
\begin{pmatrix}
#1_{0} & #1_{1} & \dots& #1_{#3-1}\\
#1_{0}^{[1]} & #1_{1}^{[1]} & \dots& #1_{#3-1}^{[1]}\\
\vdots &\vdots&\ddots& \vdots\\
#1_{0}^{[#2-1]} & #1_{1}^{[#2-1]} & \dots& #1_{#3-1}^{[#2-1]}\\
\end{pmatrix}}

\renewcommand{\a}{\mathbf a}

\renewcommand{\c}{\mathbf c}
\newcommand{\e}{\mathbf e}
\newcommand{\f}{\mathbf f}
\renewcommand{\r}{\mathbf r}
\newcommand{\s}{\mathbf s}

\newcommand{\x}{\mathbf x}

\newcommand{\A}{\mathbf A}
\newcommand{\B}{\mathbf B}
\newcommand{\C}{\mathbf C}

\newcommand{\E}{\mathbf E}
\newcommand{\G}{\mathbf G}
\newcommand{\I}{\mathbf I}
\renewcommand{\H}{\mathbf H}
\newcommand{\Q}{\mathbf Q}
\newcommand{\R}{\mathbf R}
\renewcommand{\S}{\mathbf S}

\newcommand{\X}{\mathbf X}

\newcommand{\0}{\mathbf 0}

\newcommand{\mycode}[1]{\ensuremath{\mathsf{#1}}}

\newcommand{\Gab}[1]{\ensuremath{\mycode{Gab}[#1]}}
\newcommand{\IntGab}[1]{\ensuremath{\mathrm{I}\mycode{Gab}[#1]}}
\newcommand{\IntGabstar}[1]{\ensuremath{\mathrm{I}\mycode{Gab}^*[#1]}}
\newcommand{\IntGabNoInput}{\ensuremath{\mathrm{I}\mycode{Gab}}}

\newcommand{\dhalf}{\left\lfloor (d-1)/2\right\rfloor}

\newcommand{\nkhalf}{\left\lfloor (n-k)/2\right\rfloor}

\newcommand{\nkint}{\left\lfloor s(n-k)/(s+1)\right\rfloor}

\newcommand{\Uniquecorrcap}{\ensuremath{\tau_0}}

\newcommand{\numbRowErasures}{\varrho}
\newcommand{\numbColErasures}{\gamma}

\newcommand{\myspace}[1]{\mathcal{#1}}
\newcommand{\Rowspace}[1]{\myspace{R}_q\left(#1\right)}

\newcommand{\Colspace}[1]{\myspace{C}_q\left(#1\right)}

\newcommand{\Ball}[2]{\mathcal{B}^{(#1)}(#2)}
\newcommand{\Sphere}[2]{\mathcal{S}^{(#1)}(#2)}

\newcommand{\List}{\mathcal L}

\smartqed

\begin{document}

\title{List and Unique Error-Erasure Decoding of Interleaved Gabidulin Codes with Interpolation Techniques\thanks{The work of A. Wachter-Zeh has been supported by the German Research Council (DFG) under Grant No. Bo867/21 and a Minerva Postdoctoral Fellowship. The work of A. Zeh has been supported by the German Research Council (DFG) under Grants No. Bo867/22 and Ze1016/01.\\
This work was partly presented at the International Workshop on Coding and Cryptography (WCC), Apr. 2013, Bergen, Norway \cite{WachterZeh2013InterpolationInterleavedGabidulin_conf}. 
}}
\titlerunning{Interpolation-Based Decoding of Interleaved Gabidulin Codes}
\author{Antonia Wachter-Zeh \and Alexander Zeh}
\date{\today}
\authorrunning{A.~Wachter-Zeh and A.~Zeh}
\institute{A.~Wachter-Zeh was with the Institute of Communications Engineering, University of Ulm, Ulm, Germany and the 
Institut de Recherche Mathématique de Rennes (IRMAR), Université de Rennes 1, Rennes, France and is now with the Computer Science Department, Technion---Israel Institute of Technology, Haifa, Israel.\\
\email{\texttt{antonia@codingtheory.eu}}\\
A.~Zeh was with the Institute of Communications Engineering, University of Ulm, Ulm, Germany and Research Center INRIA Saclay - \^{I}le-de-France, \'{E}cole Polytechnique, France and is now with the Computer Science Department, Technion---Israel Institute of Technology, Haifa, Israel.\\
\email{\texttt{alex@codingtheory.eu}}}

\maketitle

\begin{abstract}
A new interpolation-based decoding principle for interleaved Gabidulin codes is presented.
The approach consists of two steps: First, a multi-variate linearized polynomial is constructed which interpolates the coefficients of the received word and second, the roots of this polynomial have to be found.
Due to the specific structure of the interpolation polynomial, both steps (interpolation and root-finding) can be accomplished by solving a linear system of equations.
This decoding principle can be applied as a list decoding algorithm (where the list size is not necessarily bounded polynomially) as well as an efficient probabilistic unique decoding algorithm. 
For the unique decoder, we show a connection to known unique decoding approaches and give an upper bound on the failure probability.
Finally, we generalize our approach to incorporate not only errors, but also row and column erasures.
\end{abstract}


%

\section{Introduction}
During the last years, \emph{random linear network coding} (RLNC) has been attracting a lot of attention as a powerful means for spreading information in networks from sources to sinks \cite{Ahlswede_NetworkInformationFlow_2000,HoKoetterMedardKargerEffros-BenefitsCodingRouting_2003,MedardKoetterKargerEffrosShiLeong-RLNCMulticast_2006}. 
K\"otter and Kschischang \cite{koetter_kschischang} used subspace codes for error control in RLNC.
A subspace code is a non-empty set of subspaces of the vector space of dimension $n$ over a finite field and each codeword is a subspace itself,
compare \cite{Wang2003Linear,koetter_kschischang,Xia2009Johnson,Etzion2009ErrorCorrecting,Skachek2010Recursive,Etzion2011ErrorCorrecting,Bachoc2012Bounds}.
Silva, Kschischang and K\"otter \cite{silva_rank_metric_approach} showed that \emph{lifted Gabidulin codes} provide almost optimal subspace codes for RLNC. 
Gabidulin codes are the rank-metric analogs of Reed--Solomon codes and were introduced by Delsarte~\cite{Delsarte_1978}, Gabidulin~\cite{Gabidulin_TheoryOfCodes_1985} and Roth \cite{Roth_RankCodes_1991}. 
A \emph{lifted} Gabidulin code is a special subspace code, where each codeword is the row space of a matrix $[\,\I \ \ \C^T]$, $\I$ denotes the identity matrix and $\C$ is a codeword in matrix representation of a Gabidulin code.

\emph{Interleaved} Gabidulin codes can be seen as $s$ parallel codewords of Gabidulin codes. When applied to RLNC, they can be advantageous compared to usual Gabidulin codes since only {one} identity matrix is appended to $s$ Gabidulin codewords which reduces the relative "overhead". 
Independently from this application, it is remarkable that they can be decoded beyond the usual \emph{bounded minimum distance} (BMD) decoding capability with high probability.

In this contribution, a new interpolation-based approach for decoding interleaved Gabidulin codes of length $n$, interleaving order $s$ and elementary dimensions $k^{(i)}$, $\forall i\in \intervallincl{1}{s}$, is presented\footnote{Throughout this paper, $\intervallincl{a}{b}$ is a short-hand notation for the set of integers $ \{i:a \leq i \leq  b\}$.}.
Our decoding principle relies on constructing a multi-variate linearized polynomial which interpolates the $s$ elementary received words. 
We prove that the evaluation polynomials (of $q$-degree less than $k^{(i)}$) of any interleaved Gabidulin codeword in rank distance less than ${(sn  - \sum_{i=1}^{s}k^{(i)} +s)}/{(s+1)}$ are roots of this multi-variate polynomial.
Due to the structure of the multi-variate interpolation polynomial, its roots can be found by simply solving a linear system of equations.
This idea is related to the "linear-algebraic" decoding methods by Guruswami and Wang for folded/derivative Reed--Solomon codes \cite{Guruswami2011Linearalgebraic,GuruswamiWang-LinearAlgebraicForVariantsofReedSolomonCodes_2012} and Mahdavifar and Vardy for folded Gabidulin codes \cite{Mahdavifar2012Listdecoding}.

This paper is structured as follows.
In Section~\ref{sec:prelim}, we give notations and definitions.
Section~\ref{subsec:intgab_interpolation_algo} explains the basic principle of our decoder and shows how the two main steps---interpolation and root-finding---can each be accomplished by solving a linear system of equations. 
Our decoder is first applied as a (not necessarily polynomial-time) list decoding algorithm in Section~\ref{subsec:intgab_list_decoding} and second, as a unique decoding algorithm with a certain failure probability in Section~\ref{subsec:intgab_uniquedecoding}. 
Finally, in Section~\ref{subsec:intgab_errorerasure}, we show how our algorithm can be generalized to error-erasure decoding and conclude the paper in Section~\ref{sec:conclusion}.

\section{Preliminaries and Known Approaches}\label{sec:prelim}
\subsection{Definitions and Notations}
Let $q$ be a power of a prime, and let $\Fq$ be the finite field of order $q$ and by $\Fqm$ its extension field of degree $m$. 
We use $\Fq^{s \times n}$ to denote the set of all $s\times n$ matrices over $\Fq$ and 
$\Fqm^n =\Fqm^{1 \times n}$ for the set of all {row} vectors of length $n$ over $\Fqm$. 
Therefore, $\Fq^n$ denotes the vector space of dimension $n$ over $\Fq$. 
Denote $[i] \defeq q^i$ for any integer $i$.
For a vector $\mathbf a = \vecelements{a}\in\Fqm^n$ 
the \emph{$q$-Vandermonde matrix} is defined by
\begin{align}
\Mooremat{s}{\a} 
 \defeq \MoormatExplicit{a}{s}{n}.\label{eq:def_Moorematrix}
\end{align}
If $a_0$, $a_1, \dots$, $a_{n-1}$ are linearly independent over $\Fq$, then $\Mooremat{s}{\a}$ has rank $\min\{s,n\}$, see e.g. \cite[Lemma 3.15]{Lidl-Niederreiter:FF1996}.

A \emph{linearized polynomial}, see \cite{Ore_OnASpecialClassOfPolynomials_1933,Ore_TheoryOfNonCommutativePolynomials_1933,Lidl-Niederreiter:FF1996},
over $\Fqm$ has the form
\begin{equation*}
f(x) = \sum_{i=0}^{d_f} f_i x^{[i]}, 
\end{equation*}
with $f_i \in \Fqm$, $\forall i \in \intervallincl{0}{d_f}$. 
If $f_{d_f}\neq 0$, we call $d_f \defeq \deg_q f(x)$ the \textit{q-degree} of $f(x)$. 
For all $\alpha_1,\alpha_2 \in \Fq$ and all $a,b \in \Fqm$, it holds that
$f(\alpha_1 a+\alpha_2 b) = \alpha_1 f(a)+\alpha_2 f(b)$.
The (usual) addition and the non-commutative composition $f(g(x))$ convert the set of linearized polynomials into a non-commutative ring with the identity element $x^{[0]}=x$. In the following, all polynomials are linearized polynomials and $\Linpolyring$ denotes the ring of linearized polynomials. 
Further, for some $\vec{g} = \vecelements{g}$ and some $a(x) \in \Linpolyring$, we denote $a(\vec{g}) = \veceval{a}{g}$.

Throughout this paper, we use linearized Lagrange interpolation.
Let the elements in $\myset{G}= \setelements{g}{n}\subseteq \Fqm$ be linearly independent over $\Fq$ (as in Definition~\ref{def:int_gab}). Given $a(x) = \sum_{i=0}^{n-1}a_ix^{[i]}$, let $\qtrafo{a}(x)\in \Linpolyring$ denote the unique linearized polynomial of $q$-degree less than $n$ such that $\qtrafo{a}(g_i)=a_i$, $\forall i$, which can be calculated by:
\begin{equation}\label{eq:received_lagrange}
\qtrafo{a}(x)  = \sum\limits_{i=0}^{n-1} a_i \cdot \frac{L_i(x)}{L_i(g_i)},
\end{equation}
where $L_i(x)$ denotes the {$i$-th linearized Lagrange basis polynomial} of $q$-degree $n-1$ (see \cite{SilvaKschischang-RankMetricPriorityTransmission}), which is defined as the minimal subspace polynomial of $\myset{G}\setminus g_i = \{ g_0,\dots,g_{i-1}, g_{i+1},\dots,g_{n-1}\}$, i.e.:
\begin{equation}\label{eq:linearized_Lagrange_basis_poly}
L_i(x) 
= \prod_{B_0=0}^{q-1}\cdots \prod_{B_{i-1}=0}^{q-1}\cdot \prod_{B_{i+1}=0}^{q-1} \cdots\prod_{B_{n-1}=0}^{q-1} \Big(x-\sum_{j=0, j\neq i}^{n-1}B_j g_j\Big).
\end{equation}
Note that ${L_i(g_j)}/{L_i(g_i)} = 1$ if $i=j$ and $0$ else.

For a given basis of $\Fqm$ over $\Fq$, there is a bijective mapping for each vector $\mathbf x \in \mathbb{F}_{q^m}^n$ on a matrix $\mathbf X \in \Fq^{m \times n}$. 
Let $\rk(\mathbf x)$ denote the (usual) rank of $\mathbf X$ over $\mathbb{F}_{q}$ and let 
$\Rowspace{\X}$ and $\Colspace{\X}$ denote the row  and column space of $\X$ over $\Fq$. 
The right kernel of a matrix is denoted by $\ker(\x) = \ker(\X)$. 
The rank-nullity theorem states that for an $m \times n$ matrix, if $\dim \ker(\x) = t$, then $\dim \Colspace{\X}=\rk(\x) = n-t$.
Throughout this paper, we use the notation as vector (e.g. from $\Fqm^n$) or matrix (e.g. from $\Fq^{m \times n}$) equivalently, whatever is more convenient.

The \emph{minimum rank distance} $d$ of a block code $\mycode{C}$ over $\Fqm$ is defined by 
\begin{equation*}
d \defeq \min_{\substack{\mathbf c_1,\c_2 \in \mycode{C}\\ \mathbf c_1 \neq \mathbf c_2}} \rk(\mathbf c_1-\c_2) . 
\end{equation*}
%
%
%
For linear codes of length $n \leq m$ and dimension $k$, the Singleton-like upper bound \cite{Delsarte_1978,Gabidulin_TheoryOfCodes_1985,Roth_RankCodes_1991} implies that $d \leq n-k+1$.
If $d=n-k+1$, the code is called a \emph{maximum rank distance} (MRD) code. 

Further, $\Ball{\tau}{\a}$ denotes a ball of radius $\tau$ in rank metric around a word $\a\in \Fqm^n$ and $\Sphere{\tau}{\a}$ denotes a sphere in rank metric of radius $\tau$ around the word $\a$. 

\subsection{(Interleaved) Gabidulin Codes}
Gabidulin codes \cite{Delsarte_1978,Gabidulin_TheoryOfCodes_1985,Roth_RankCodes_1991} are special MRD codes and 
are considered as rank-metric analogs of Reed--Solomon codes.
%
Interleaved Gabidulin codes consist of $s$ horizontally or vertically arranged codewords of (not necessarily different) Gabidulin codes.
Vertically interleaved Gabidulin codes were introduced by Loidreau and Overbeck in 
\cite{Loidreau_Overbeck_Interleaved_2006,Overbeck_Diss_InterleveadGab} and 
rediscovered by Silva, Kschischang and Kötter \cite{silva_rank_metric_approach,Silva_PhD_ErrorControlNetworkCoding} as the 
{Cartesian product} of $s$ transposed codewords of Gabidulin codes. 
Later, Sidorenko and Bossert introduced \emph{horizontally} interleaved Gabidulin codes 
\cite{SidBoss_InterlGabCodes_ISIT2010,Sidorenko2011SkewFeedback}. 



We consider \emph{vertically} interleaved Gabidulin codes.
However, if one requires matrices with the dimensions of a horizontally interleaved Gabidulin code, we can simply transpose all codewords. 


\begin{definition_it}[Interleaved Gabidulin Code]\label{def:int_gab}
A linear (vertically) interleaved\\ Gabidulin code $\IntGab{s;n,k^{(1)},\dots,k^{(s)}}$ over $\Fqm$ of length $n \leq m$, elementary dimensions $k^{(1)},\dots,k^{(s)} \leq n$, and interleaving order $s$ is defined by
\begin{equation*}
\IntGab{s;n,k^{(1)},\dots,k^{(s)}}\defeq 
\left\lbrace
\begin{pmatrix}
f^{(1)}(\vec{g})\\
f^{(2)}(\vec{g})\\
\vdots\\
f^{(s)}(\vec{g})\\
\end{pmatrix}
 :  \deg_q f^{(i)}(x) < k^{(i)} \leq n, \forall i \in \intervallincl{1}{s}
\right\rbrace,
\end{equation*}
where $f^{(i)}(x)\in \Linpolyring$, $\forall i \in \intervallincl{1}{s}$, $\vec{g} = \vecelements{g}$ and the fixed elements $g_0,g_1, \dots, g_{n-1} \in \Fqm$ are linearly independent over $\Fq$. 
\end{definition_it}
Note that $\vec{c}^{(i)} = f^{(i)}(\vec{g}) \in \Gab{n,k^{(i)}}= \IntGab{1;n,k^{(i)}}$.
We can represent the codewords of the interleaved code as matrix in $\Fqm^{s\times n}$ or as matrix in $\Fq^{sm \times n}$.

\begin{corollary} 
Let \IntGab{s;n,k,\dots,k} be a linear interleaved Gabidulin code over $\Fqm$ as in Definition~\ref{def:int_gab} with $k^{(i)} = k$, $\forall i\in\intervallincl{1}{s}$.
Its minimum rank distance is $d=n-k+1$ and it is an MRD code.
\end{corollary}
In general, for arbitrary $k^{(i)}$, the minimum rank distance of \IntGab{s;n,k^{(1)},\dots,k^{(s)}} is $d = n-\max_i\{k^{(i)}\}+1$, which is not necessarily an MRD code.

\subsection{Known Approaches for Decoding Interleaved Gabidulin Codes}\label{sec:joint_decoding}

So far, there are two approaches for decoding interleaved Gabidulin codes: \cite{Loidreau_Overbeck_Interleaved_2006}
and \cite{SidBoss_InterlGabCodes_ISIT2010}.
Both are probabilistic unique decoding algorithms up to the radius $\tau = \nkint$ (for $k^{(i)}=k$, $\forall i\in\intervallincl{1}{s}$) and return the unique solution with high probability.
In the following, we shortly summarize the two principles and prove a relation between them. It is important to remark that the approach from \cite{SidBoss_InterlGabCodes_ISIT2010} was originally described for \emph{horizontally} interleaved Gabidulin codes, i.e., where an interleaved codeword is defined by $(f^{(1)}(\vec{g}) \ {f}^{(2)}(\vec{g})\ \dots \ {f}^{(s)}(\vec{g}))$, but in the following, we describe it for \emph{vertically} interleaved Gabidulin codes as in Definition~\ref{def:int_gab}. 

Let $\r^{(i)} = \vecelements{r^{(i)}}$, $\forall i \in \intervallincl{1}{s}$, denote the $s$ elementary received words, i.e., 
$\r^{(i)} = \c^{(i)}+ \e^{(i)}$ and $\c^{(i)} \in \Gab{n,k^{(i)}}$ as in Definition~\ref{def:int_gab}.
Further, let $t^{(i)} = \rk(\e^{(i)})$ and let $t \defeq \rk (\e^{(1)T} \ \e^{(2)T} \ \dots \ \e^{(s)T} )$.
We assume throughout this paper that every matrix $(\e^{(1)T} \ \e^{(2)T} \ \dots \ \e^{(s)T} )^T \in \Fqm^{s \times n}$ of rank $t$ is {equi-probable}.


For the explanation of the two known decoding principles, we assume that we know the actual rank of the error $t$, which enables us to directly set up the corresponding system of equations with the appropriate size.
A straight-forward algorithmic realization would therefore solve this system of equations for every $t$, where $\dhalf+1 \leq t \leq \tau$, but this principle can easily be improved.
%

\subsubsection*{A Decoding Approach based on the Received Word}
We show the main properties of the algorithm from \cite{Loidreau_Overbeck_Interleaved_2006} in the following;
for details the reader is referred to \cite{Loidreau_Overbeck_Interleaved_2006,Overbeck_Diss_InterleveadGab,Overbeck_AttacksPublicKeyCryptoGabidulin_2008}.
For some $t \leq \tau$, the main step of the decoding algorithm from \cite{Loidreau_Overbeck_Interleaved_2006} is to solve a homogeneous linear system of equations
\begin{equation}\label{eq:intgab_lo_systemeq}
\R_{R} \cdot \boldsymbol{\lambda}^T = \0,
\end{equation}
for $\boldsymbol{\lambda} = \vecelements{\lambda}$, where the $(n-t-1 + s(n-t)-\sum_{i=1}^{s}k^{(i)})\times n$ matrix $\R_{R}$ depends on $\vec{g}=\vecelements{g}$ and the received words:
\begin{equation}\label{eq:matrix_loidreauoverbeck}
\R_{R} = 
\begin{pmatrix}
\G_{R}\\
\R_{R}^{(1)}\\
\R_{R}^{(2)}\\
\vdots\\
\R_{R}^{(s)}\\
\end{pmatrix}
\defeq
\begin{pmatrix}
\Mooremat{n-t-1}{\vec{g}}\\
\Mooremat{n-k^{(1)}-t}{\r^{(1)}}\\
\Mooremat{n-k^{(2)}-t}{\r^{(2)}}\\
\vdots\\
\Mooremat{n-k^{(s)}-t}{\r^{(s)}}\\
\end{pmatrix},
\end{equation}
and "qvan" defines the $q$-Vandermonde matrix as in \eqref{eq:def_Moorematrix}.

If the right kernel of $\R_{R}$ has dimension one, the closest interleaved codeword can be reconstructed, 
see \cite{Loidreau_Overbeck_Interleaved_2006} and \cite[Algorithm~3.2.1]{Overbeck_Diss_InterleveadGab}.
When $\R_{R}$ has rank less than $n-1$, the codeword cannot be reconstructed in most cases. 
Thus, the \emph{decoding failure} is at most the probability that $\rk(\R_{R})$ is less than $n-1$.


The first $k^{(i)}$ rows of $\G_R$, for $i \in \intervallincl{1}{s}$, constitute the generator matrix of the $\Gab{n,k^{(i)}}$ code, which is the $i$-th elementary code of the $\IntGab{s;n,k^{(1)},\dots,k^{(s)}}$ code. 
This is due to the fact that $t \leq \tau \leq n-\max_i\{k^{(i)}\}-1$, and hence, $k^{(i)} \leq n-t -1$, $\forall i \in \intervallincl{1}{s}$. 
Therefore, the right kernel of $\R_{R}$ can also be expressed in terms of the elementary {error} words:
\begin{equation}\label{eq:intab_lo_kernel}
\ker\big(\R_{R}\big) = 
\ker
\begin{pmatrix}
\Mooremat{n-t-1}{\vec{g}}\\
\Mooremat{n-k^{(1)}-t}{\e^{(1)}}\\
\vdots\\
\Mooremat{n-k^{(s)}-t}{\e^{(s)}}\\
\end{pmatrix}
\defeq
\ker\left(\E_{R}\right).
\end{equation}
The rank of $\G_R$ is $n-t-1$ and the rank of the lower $s$ submatrices of $\E_{R}$ is $t \leq \tau$. 
Hence, the overall rank is $\rk(\R_{R}) = \rk(\E_{R})\leq n-1$. 
For $s \leq t $, the probability that $\rk(\R_{R}) < n-1$ is upper bounded in \cite[Eq.~(6)]{Loidreau_Overbeck_Interleaved_2006}, \cite[Eq.~(12)]{Overbeck_Diss_InterleveadGab} as follows:
\begin{equation}\label{eq:intgab_failure_lo}
P\big(\rk(\R_{R}) < n-1\big) \leq 1 - \left(1-\frac{4}{q^m}\right)\left(1-q^{m(s-t)}\right)^s.
\end{equation}

\subsubsection*{A Syndrome-Based Decoding Approach}
The approach from \cite{SidBoss_InterlGabCodes_ISIT2010,Sidorenko2011SkewFeedback} 
is a generalization of key equation-based decoding of Gabidulin codes \cite{Gabidulin_TheoryOfCodes_1985,Roth_RankCodes_1991}. 
Denote $s$ syndrome vectors of length $n-k^{(i)}$ by:
\begin{equation*}
\s^{(i)} \defeq \r^{(i)}\cdot \H^{(i)T} =  \e^{(i)}\cdot \H^{(i)T} = (s_0^{(i)} \ s_1^{(i)} \ \dots \ s_{n-k^{(i)}-1}^{(i)}), \quad \forall i \in \intervallincl{1}{s},
\end{equation*}
where $\H^{(i)}$ is a parity-check matrix of the elementary $\Gab{n,k^{(i)}}$ code, $\forall i \in \intervallincl{1}{s}$.
Further, we define $s$ modified syndromes by the following coefficients:
\begin{equation*}
\widetilde{s}^{(i)}_j  = s^{(i)[j-n+k^{(i)}+1]}_{n-k^{(i)}-1-j}, \forall i \in \intervallincl{1}{s}, \forall j \in \intervallexcl{0}{n-k^{(i)}},
\end{equation*}
and denote the corresponding polynomials by $\widetilde{s}^{(i)}(x) = \sum_{j=0}^{n-k^{(i)}-1}\widetilde{s}^{(i)}_jx^{[j]}$.
 Then, a \emph{key equation} for the row space of the whole error matrix holds as follows:
\begin{equation}\label{eq:interleavedGabi_SidBoss}
\renewcommand{\arraystretch}{1.3}
\S\cdot \boldsymbol{\Gamma}^T =\begin{pmatrix}
\S^{(1)}\\
\S^{(2)}\\
\vdots\\
\S^{(s)}\\
\end{pmatrix} \cdot \boldsymbol{\Gamma}^T = \0,
\end{equation}
where $\boldsymbol{\Gamma} = (\Gamma_0 \ \Gamma_{1} \ \dots \ \Gamma_t)$ and 
\begin{equation}\label{eq:interleavedGabi_matrix_SidBoss}
\S^{(i)}
= 
\renewcommand{\arraystretch}{1.7}
\setlength{\arraycolsep}{.8ex}
\begin{pmatrix}
s_{n-k^{(i)}-1-t}^{(i)[t-n+k^{(i)}+1]}& s_{n-k^{(i)}-t}^{(i)[t-n+k^{(i)}+1]} & \dots &s_{n-k^{(i)}-1}^{(i)[t-n+k^{(i)}+1]} \\
s_{n-k^{(i)}-2-t}^{(i)[t-n+k^{(i)}+2]}& s_{n-k^{(i)}-t-1}^{(i)[t-n+k^{(i)}+2]} & \dots &s_{n-k^{(i)}-2}^{(i)[t-n+k^{(i)}+2]} \\
\vdots &\vdots&\ddots& \vdots\\
s_{0}^{(i)[0]}& s_{1}^{(i)[0]} & \dots &s_{t}^{(i)[0]} \\
\end{pmatrix}, \ \forall i \in \intervallincl{1}{s}.\renewcommand{\arraystretch}{1}
\end{equation}
If $\rk(\S) =t$, we obtain a unique solution of the error span polynomial $\Gamma(x)$ (except for a scalar factor) and we can reconstruct the $s$ error vectors, compare \cite{SidBoss_InterlGabCodes_ISIT2010}.

Hence, the probability of failure for the approach from \cite{SidBoss_InterlGabCodes_ISIT2010} can be upper bounded by the probability that $\S$ from \eqref{eq:interleavedGabi_SidBoss} has rank less than $t$, which is bounded in \cite[Theorem~5]{SidBoss_InterlGabCodes_ISIT2010} for $s \leq \tau$ by:
\begin{equation*}
P\big(\rk(\S) < t\big)\leq 3.5 \; q^{-m\big((s+1)(\tau-t)+1\big)}  < \frac{4}{q^m}. 
\end{equation*}
This bound improves the bound from \cite{Loidreau_Overbeck_Interleaved_2006} and in general we can use $P_f < {4}/{q^m}$ as simplified upper bound on the failure probability of both cases.


\subsubsection*{Connection Between the Two Known Approaches}
\begin{lemma}[Relation Between Decoding Matrices]\label{lem:intgab_rel_LO_SB}
Let $k^{(i)} = k$, $\forall i \in \intervallincl{1}{s}$, let $t \leq \tau = \nkint$ and let $\R_{R}$ be defined as in \eqref{eq:matrix_loidreauoverbeck} and $\S$ as in \eqref{eq:interleavedGabi_SidBoss}, \eqref{eq:interleavedGabi_matrix_SidBoss}.
Then, $\rk(\S) < t$ if and only if $\rk(\R_{R}) < n-1$.
\end{lemma}
\begin{proof}
First recall $\R_{R}$ from \eqref{eq:matrix_loidreauoverbeck}.
The submatrix $\G_R$ is a generator matrix of a $\Gab{n,n-t-1}$ code. 
Let $\vec{h} = \vecelements{h}$ define an $(n-k)\times n$ parity-check matrix $\H^{(0)}$ of the $\Gab{n,k}$ code (which defines the $\IntGab{s;n,k,\dots,k}$ code).

Then, $\H = \Mooremat{t+1}{\vecelements{h^{[n-k-t-1]}}}$ is a $(t+1)\times n$ parity check matrix of a $\Gab{n,n-t-1}$ code and is a $(t+1) \times n$ submatrix of $\H^{(0)}$, consisting of the lowermost $t+1$ rows of $\H^{(0)}$.
Multiplying $\R_{R}$ by ${\H}^{T}$ and comparing the result to \eqref{eq:interleavedGabi_matrix_SidBoss} gives:
\begin{equation}\label{eq:intgab_proof_connection}
\renewcommand{\arraystretch}{1.3}
\R_{R}  \H^T =
\begin{pmatrix}
\Mooremat{n-t-1}{\vec{g}}\\
\Mooremat{n-k-t}{\r^{(1)}}\\
\vdots\\
\Mooremat{n-k-t}{\r^{(s)}}\\
\end{pmatrix}
{\H}^T = 
\begin{pmatrix}
\Mooremat{n-t-1}{\vec{g}}\\
\Mooremat{n-k-t}{\e^{(1)}}\\
\vdots\\
\Mooremat{n-k-t}{\e^{(s)}}\\
\end{pmatrix}
{\H}^T = 
\begin{pmatrix}
\0\\
\S^{(1)[n-k-t-1]}\\
\vdots\\
\S^{(s)[n-k-t-1]}
\end{pmatrix}\!.
\renewcommand{\arraystretch}{1}
\end{equation}
For any integer $i$, $\rk(\A) = \rk(\A^{[i]})$, where $\A^{[i]}$ means that every entry is taken to the $q$-power $i$. 

Based on \eqref{eq:intgab_proof_connection}, we first prove the {if part}.
Calculate by Gaussian elimination of $\R_{R}$ the matrix  
$\widetilde{\E} = 
\begin{pmatrix}
\G_R\\
\widetilde{\E}_{R}
\end{pmatrix}$
such that $\rk(\R_{R})=\rk(\widetilde{\E}) = \rk(\G_R) + \rk(\widetilde{\E}_{R}) = n-t-1 + \rk(\widetilde{\E}_{R})$ (i.e., such that the ranks sum up).
Notice that $\widetilde{\E}_{R}$ does not necessarily consist of the $s$ lower submatrices of $\E_{R}$ from \eqref{eq:intab_lo_kernel}.
These elementary row operations do not change the rank and we obtain from \eqref{eq:intgab_proof_connection}
\begin{equation*}
\rk\big(\S\big) =\rk\big(\S^{[n-k-t-1]}\big) =  \rk\big(\R_{R} \cdot \H^T\big) =  \rk\big(\widetilde{\E} \cdot \H^T\big) = \rk\big(\widetilde{\E}_{R} \cdot \H^T\big).
\end{equation*}
Now, if $\rk(\R_{R}) <n-1$, then $\rk(\widetilde{\E}_{R}) < t$ since $ \rk(\G_R) = n-t-1  $.
Then, also $\rk\big(\widetilde{\E}_{R} \cdot \H^T\big) <t$ and therefore $\rk(\S) <t$.\\
Second, let us prove the {only if part}.
Due to Sylvester's rank inequality
\begin{equation*}
\rk\big(\R_{R}\big)+ \rk\big(\H^T\big)-n \leq \rk\big(\R_{R} \cdot \H^T\big)  = \rk\big(\S^{[n-k-t-1]}\big) = \rk\big(\S\big).
\end{equation*}
Clearly, $\rk(\H) = t+1$. Hence, if $\rk(\S) <t$, then
$\rk\big(\R_{R}\big) \leq n-t-1 + \rk\big(\S\big) < n-1$.
\qed\end{proof}
Thus, both approaches have the same fraction of correctable error matrices when $k^{(i)} = k$, $\forall i\in \intervallincl{1}{s}$.
This means that the tighter bound on the failure probability from \cite{SidBoss_InterlGabCodes_ISIT2010} can also be used to bound the failure probability of \cite{Loidreau_Overbeck_Interleaved_2006}.

However, for arbitrary $k^{(i)}$, it is not clear if the matrix on the RHS of \eqref{eq:intgab_proof_connection} has the same rank as $\S$ since the $q$-powers of each submatrix differ.

\section{Principle of Interpolation-Based Decoding}\label{subsec:intgab_interpolation_algo}
Guruswami and Sudan \cite {Sudan:JOC1997, Guruswami-Sudan:IEEE_IT1999} introduced polynomial-time list decoding of Reed--Solomon and Algebraic-Geometry codes based on interpolating bivariate (usual) polynomials.
For \emph{linearized} polynomials, however, it is not clear how to define mixed terms (i.e., monomials containing more than one indeterminate) and how to design a list decoding algorithm for Gabidulin codes, see also \cite{Wachterzeh_BoundsListDecodingRankMetric_IEEE-IT_2013}.
When a bivariate linearized polynomials is defined \emph{without} mixed terms, we can decode an $\Gab{n,k}$ code up to $\nkhalf=\dhalf$, which was done in \cite {Loidreau_AWelchBerlekampLikeAlgorithm_2006}.

Our decoding approach for \emph{interleaved} Gabidulin codes is based on interpolating a \emph{multi}-variate linearized polynomial without mixed terms. 

\subsection{Interpolation Step}\label{subsec:interpolation}
\begin{problem_bold}[Interpolation Step]\label{prob:interpolation}
Let $r^{(i)}(x) = \sum_{j=0}^{n-1}r_j^{(i)}x^{[j]} \in \Linpolyring$, $\forall i \in \intervallincl{1}{s}$, and $g_0,g_1,\dots,g_{n-1} \in \Fqm$, which are linearly independent over $\Fq$, be given.

Find an $(s+1)$-variate linearized polynomial of the form 
\begin{equation*}
Q(x,y_1,\dots,y_s) = Q_0(x) + Q_1(y_1) + \dots + Q_s(y_s), 
\end{equation*}
which satisfies for given integers $n,\tau,k^{(1)}, \dots, k^{(s)}$:
\begin{itemize}
\item[$\bullet$] $Q(g_j,r_j^{(1)},\dots,r_j^{(s)}) = 0$, $\forall j \in \intervallexcl{0}{n}$,
\item[$\bullet$] $\deg_q Q_0(x) < n-\tau$,
\item[$\bullet$] $\deg_q Q_i(y_i)< n-\tau-(k^{(i)}-1)$, $\forall i \in \intervallincl{1}{s}$.
\end{itemize}
\end{problem_bold}
Denote the coefficients of the univariate linearized polynomials by 
\begin{equation}\label{eq:coeffs_multivpoly}
Q_0(x) = \sum_{j=0}^{n-\tau-1} q_{0,j}x^{[j]}, \quad Q_i(y_i) = \sum_{j=0}^{n-\tau-k^{(i)}} q_{i,j}y_i^{[j]}, \quad \forall i\in \intervallincl{1}{s}.
\end{equation}
A solution to Problem~\ref{prob:interpolation} can be found by solving a linear system of equations, which is denoted by $\Mat{R} \cdot \vec{q}^T = \0$, where $\vec{g} = \vecelements{g}$ and
$\R$ is an $ n \times \big(n-\tau + \sum_{i=1}^{s}(n-\tau-k^{(i)}+1)\big)$ matrix as follows:
\begin{equation}\label{eq:interpolation_matrix}
\Mat{R} = \left( \Mooremat{n-\tau}{\vec{g}}^T \ \Mooremat{n-\tau-k^{(1)}+1}{\r^{(1)}}^T\ \dots \ \Mooremat{n-\tau-k^{(s)}+1}{\r^{(s)}}^T \right),
\end{equation} 
and $\vec{q} = (q_{0,0} \ \dots \ q_{0,n-\tau-1} \;|\; q_{1,0} \ \dots \ q_{1,n-\tau-k^{(1)}}\;| \;\dots \;|\;q_{s,0} \ \dots \ q_{s,n-\tau-k^{(s)}})$.

 \begin{lemma}\label{lem:nonzero_interpolpoly}
There is a non-zero $Q(x,y_1,\dots,y_s)$, fulfilling the conditions of Problem~\ref{prob:interpolation} if
 \begin{equation}\label{eq:decradius_het}
 \tau < \frac{sn  - \sum_{i=1}^{s}k^{(i)} +s}{s+1}.
 \end{equation}
 \end{lemma}
 \begin{proof}
 The number of linearly independent equations is at most the number of interpolation constraints (i.e., the number of rows of $\R$ in \eqref{eq:interpolation_matrix}), i.e., $n$, and has to be less than the number of unknowns (given by the length of $\vec{q}$) in order to guarantee that there is a non-zero solution:
 \begin{equation*}
   n < n-\tau + \sum\limits_{i=1}^{s} \left( n-\tau-k^{(i)}+1\right)
  \quad \Longleftrightarrow \quad \tau (s+1) < sn +s - \sum\limits_{i=1}^{s} k^{(i)}.
\end{equation*} 
\qed\end{proof}
For the special case $k^{(i)}=k$, $\forall i\in \intervallincl{1}{s}$, this gives  $\tau < s(n-k+1)/(s+1)$.
 
 The unique decoding approaches from \cite {Loidreau_Overbeck_Interleaved_2006,SidBoss_InterlGabCodes_ISIT2010} (see Section~\ref{sec:prelim}) have maximum decoding radius $\tau_{u} = \lfloor{(sn - \sum_{i=1}^{s} k^{(i)})}/{(s+1)}\rfloor$. A comparison to the maximum value of $\tau$, given by Lemma~\ref{lem:nonzero_interpolpoly}, provides the following corollary and shows that our decoding radius is at least the same as $\tau_{u}$.
 \begin{corollary}
 Let $\tau_{u} =  \lfloor {(sn - \sum_{i=1}^{s} k^{(i)})}/{(s+1)} \rfloor$ and let $\tau$ be the greatest integer fulfilling \eqref{eq:decradius_het}. 
 Then, 
 $1\geq\tau - \tau_{u} \geq 0$.
 \end{corollary}
 


The following theorem shows that the evaluation words of the interleaved Gabidulin code are a root of any valid interpolation polynomial. 
 \begin{theorem}[Roots of Interpolation Polynomial]\label{theo:interpolpolyiszero}
Let $\c^{(i)} = f^{(i)}(\vec{g})$, where\\ $\deg_q f^{(i)}(x)$ $ < k^{(i)}$, and let $\r^{(i)}  = \c^{(i)}+\e^{(i)}$, $\forall i \in \intervallincl{1}{s}$. 

Let $t= \rk \big(\e^{(1)T} \ \e^{(2)T} \ \dots \ \e^{(s)T} \big) \leq \tau$,
where $\tau$ satisfies \eqref{eq:decradius_het}.
  Let $Q(x,y_1,\dots,y_s) \neq 0$ be given, fulfilling the interpolation constraints from Problem~\ref{prob:interpolation}. Then,
 \begin{equation}\label{eq:defi_px_null}
 F(x) \defeq Q\left(x,f^{(1)}(x), \dots,f^{(s)}(x)\right) = 0.
 \end{equation}
  \end{theorem}
\begin{proof}
Define $\qtrafo{r}^{(i)}(x)$ and $\qtrafo{e}^{(i)}(x)$ such that $\qtrafo{r}^{(i)}(g_j) = r_j^{(i)}$ and $\qtrafo{e}^{(i)}(g_j) = e_j^{(i)}=r_j^{(i)}-c_j^{(i)}$, $\forall j\in \intervallexcl{0}{n}$ and $\forall i \in \intervallincl{1}{s}$ 
as in \eqref{eq:received_lagrange}, \eqref{eq:linearized_Lagrange_basis_poly}. 
Further, denote $R(x) \defeq Q\big(x,\qtrafo{r}^{(1)}(x), \dots,\qtrafo{r}^{(s)}(x)\big)$. 
Since all polynomials are linearized, 
\begin{align*}
&R(x) - F(x) =\\ 
&Q\big(0,\qtrafo{e}^{(1)}(x), \dots,\qtrafo{e}^{(s)}(x)\big) = Q_1\big(\qtrafo{e}^{(1)}(x)\big)+ Q_2\big(\qtrafo{e}^{(2)}(x)\big)+ \dots + Q_s\big(\qtrafo{e}^{(s)}(x)\big).
\end{align*}
Then, $R\big(\vec{g}\big) -F\big(\vec{g}\big) =$
\begin{align*}
&\sum\limits_{i=1}^{s} Q_i\big(\qtrafo{e}^{(i)}(\vec{g})\big) = \sum\limits_{i=1}^{s} Q_i\big(\e^{(i)}\big) = 
\Big(\sum\limits_{i=1}^{s} Q_i(e^{(i)}_0) \ \sum\limits_{i=1}^{s} Q_i(e^{(i)}_1) \ \dots \ \sum\limits_{i=1}^{s} Q_i(e^{(i)}_{n-1})\Big).
\end{align*}
Lemma~\ref{lem:eval_row_space} in the appendix shows that the row spaces fulfill 
\begin{equation*}
\Rowspace{\sum_{i=1}^{s} Q_i\big(\e^{(i)}\big)} \subseteq \Rowspace{(\e^{(1)T} \ \e^{(2)T} \ \dots \ \e^{(s)T})^T}.
\end{equation*}
Because of the interpolation constraints, we obtain $R(\vec{g}) = \0$ and hence, $\rk\left(F(\vec{g})\right) = \rk(\sum_{i=1}^{s} Q_i(\e^{(i)})) \leq \rk(\e^{(1)T} \ \e^{(2)T} \ \dots \ \e^{(s)T} )=t \leq \tau$. 

If $\rk(F(\vec{g})) \leq \tau$, the dimension of the root space of $F(x)$ in $\Fqm$ has to be at least $n-\tau$, which is only possible if its $q$-degree is at least $n-\tau$. However, $\deg_q F(x) \leq n-\tau-1$ due to the interpolation constraints and therefore $F(x) = 0$.
\qed\end{proof}

The interpolation step can be accomplished by solving the linear system of equations based on the matrix $\R$ from \eqref{eq:interpolation_matrix}, which requires cubic complexity in $\Fqm$ with Gaussian elimination. 
Instead of this, it seems that the efficient linearized interpolation from \cite {Xie2011General} can be used and the complexity of the interpolation step can be reduced to $\mathcal O(s^2 n (n-\tau))$ operations over $\Fqm$. 


\subsection{Root-Finding Step}\label{subsec:rootfinding}

Given $Q(x,y_1,\dots,y_s)$, fulfilling the constraints of Problem~\ref{prob:interpolation}, the task of the root-finding step is to find all tuples $(f^{(1)}(x), f^{(2)}(x),\dots, f^{(s)}(x))$ such that
\begin{equation*}
 F(x) = Q_0(x) +Q_1\big(f^{(1)}(x)\big)+Q_2\big(f^{(2)}(x)\big) + \dots + Q_s\big(f^{(s)}(x)\big) = 0.
\end{equation*}
The important observation is that this is a linear system of equations over $\Fqm$ in the coefficients of $f^{(1)}(x), f^{(2)}(x),\dots, f^{(s)}(x)$. This is similar to the root-finding step of Guruswami and Wang for folded/derivative Reed--Solomon codes \cite{Guruswami2011Linearalgebraic,GuruswamiWang-LinearAlgebraicForVariantsofReedSolomonCodes_2012} and to Mahdavifar and Vardy for folded Gabidulin codes \cite{Mahdavifar2012Listdecoding}. 
Recall for this purpose that $(a +b )^{[i]} = a^{[i]} + b^{[i]}$ for any $a,b \in \Fqm$ and any integer $i$.

\begin{example_bold}[Root-Finding]
Let $s=2$, $n=m=7$, $k^{(1)}=k^{(2)}=2$ and $\tau = 3$. Find all pairs $(f^{(1)}(x), f^{(2)}(x))$ with $\deg_q f^{(1)}(x),\deg_q f^{(2)}(x)<2$ such that $F(x)  = F_0 x^{[0]} +F_1 x^{[1]}+\dots+F_{n-\tau-1} x^{[n-\tau-1]} = 0$. 
Due to the constraints of Problem~\ref{prob:interpolation}, $\deg_q F(x) \leq n-\tau-1 = 3$. Thus, 
\begin{align*}
F_0 &= 0 = q_{0,0} + q_{1,0} f^{(1)}_0+ q_{2,0} f^{(2)}_0,\\
F_1 &= 0 = q_{0,1} + q_{1,1} f_0^{(1)[1]} + q_{1,0} f^{(1)}_1+ q_{2,1} f_0^{(2)[1]} + q_{2,0} f^{(2)}_1,\\
F_2 & = 0 = q_{0,2} + q_{1,2} f_0^{(1)[2]} + q_{1,1} f^{(1)[1]}_1+ q_{2,2} f_0^{(2)[2]} + q_{2,1} f^{(2)[1]}_1,\\
F_{3} &= 0 = q_{0,3} + q_{1,2} f_{1}^{(1)[2]} + q_{2,2} f_{1}^{(2)[2]}.
\end{align*}
Therefore, given $Q(x,y_1,y_2)$, we can calculate the coefficients of all possible pairs $f^{(1)}(x)$, $f^{(2)}(x) $ of $q$-degree less than two by the following linear system of equations:
\begin{equation}\label{eq:factorization_system_s2}
\begin{pmatrix}
q_{1,0} & q_{2,0}&	&\\
q_{1,1}^{[-1]} & q_{2,1}^{[-1]} & q_{1,0}^{[-1]}&q_{2,0}^{[-1]}\\
q_{1,2}^{[-2]} & q_{2,2}^{[-2]} & q_{1,1}^{[-2]}&q_{2,1}^{[-2]}\\
&&q_{1,2}^{[-3]} & q_{2,2}^{[-3]}\\
\end{pmatrix}
\cdot
\begin{pmatrix}
f^{(1)}_0\\
f^{(2)}_0\\
f_1^{(1)[-1]}\\
f_1^{(2)[-1]}\\
\end{pmatrix}
=
\begin{pmatrix}
-q_{0,0}\\
-q_{0,1}^{[-1]}\\
-q_{0,2}^{[-2]}\\
-q_{0,3}^{[-3]}\\
\end{pmatrix}.
\end{equation}
\end{example_bold}
In order to set up \eqref{eq:factorization_system_s2} in general, we can use more than one $Q(x,y_1,\dots,y_s)$. Namely, we can use all polynomials corresponding to different basis vectors of the solution space of the interpolation step. This also decreases the probability that the system of equations for the root-finding step does not have full rank (see also Section~\ref{subsec:intgab_uniquedecoding}). 
In order to calculate the dimension of the solution space of the interpolation step, denoted by $d_I$, we need the rank of the interpolation matrix.
\begin{lemma} 
Let $\rk\big(\e^{(1)T} \ \e^{(2)T} \ \dots \ \e^{(s)T} \big) = t \leq \tau$,
where $\tau$ satisfies \eqref{eq:decradius_het}.
Then, for the interpolation matrix from \eqref{eq:interpolation_matrix}, $\rk(\R) \leq n-\tau + t$ holds.
\end{lemma}
\begin{proof}
The first $k^{(i)}$ columns of $\Mat{R}$ contain the generator matrices of the Gabidulin codes $\Gab{n,k^{(i)}}$. For calculating the rank of $\Mat{R}$, we can subtract the codewords and their $q$-powers from the $s$ right submatrices such that
these submatrices only depend on the error. 
Hence, the rank of $\Mat{R}$ depends on $\rk(\Mooremat{n-\tau}{\vec{g}})$, which is $n-\tau$, and on the rank of the error matrix, which is $t$. 
Hence, $\rk(\Mat{R}) \leq n-\tau + t$.
\qed\end{proof}
The dimension of the solution space of the interpolation step is therefore:
\begin{align}
d_I \defeq \dim \ker (\R) &\geq (s+1)(n-\tau) - \sum\limits_{i=1}^{s}(k^{(i)}-1)-(n-\tau+t)\nonumber\\
 &= s(n-\tau+1)-\sum\limits_{i=1}^{s}k^{(i)}-t,\label{eq:intgab_dimsolspace_interpolation}
\end{align}
and for $k^{(i)} = k$, $\forall i \in \intervallincl{1}{s}$, we obtain
$d_I \geq s(n-\tau -k+1)-t$.

In the following, let $Q^{(h)}(x,y_1,\dots,y_s)$, $\forall h\in\intervallincl{1}{d_I}$, denote the interpolation polynomials corresponding to different basis vectors of the solution space of the interpolation step. 
We denote the following matrices:
\begin{equation}\label{eq:intgab_notation_Qij}
\Q_j^{[i]} \defeq \begin{pmatrix}
q^{(1)[i]}_{1,j} &q^{(1)[i]}_{2,j} & \dots & q^{(1)[i]}_{s,j}\\
q^{(2)[i]}_{1,j} &q^{(2)[i]}_{2,j} & \dots & q^{(2)[i]}_{s,j}\\
\vdots &\vdots&\ddots& \vdots\\
 q^{(d_I)[i]}_{1,j} &  q^{(d_I)[i]}_{2,j} &\dots & q^{(d_I)[i]}_{s,j}\\
\end{pmatrix},
\quad
\f_j^{[i]} \defeq
\begin{pmatrix}
f_j^{(1)[i]}\\
f_j^{(2)[i]}\\
\vdots\\
f_j^{(s)[i]}\\
\end{pmatrix}, 
\quad
\vec{q}_{0,j}^{[i]} \defeq
\begin{pmatrix}
q_{0,j}^{(1)[i]}\\
q_{0,j}^{(2)[i]}\\
\vdots\\
q_{0,j}^{(d_I)[i]}\\
\end{pmatrix}. 
\end{equation}
The linear system of equations for finding the roots of $Q(x,y_1,\dots,y_s)$, where $k = \max_i \{k^{(i)}\}$, is: 
\begin{align}
&Q^{(h)}\big(x,f^{(1)}(x), \dots,f^{(s)}(x)\big) =\\
&\hspace{8ex}Q_0^{(h)}(x) +Q_1^{(h)}(f^{(1)}(x)) + \dots + Q_s^{(h)}(f^{(s)}(x)) = 0, \quad \forall h \in \intervallincl{1}{d_I} \nonumber\\
 & \hspace{30ex}\Longleftrightarrow \nonumber\\
&\renewcommand{\arraystretch}{1.4}
\begin{pmatrix}
\Q_0^{[0]}\\
\Q_1^{[-1]} & \Q_0^{[-1]}\\
\Q_2^{[-2]} &\Q_1^{[-2]} & \Q_0^{[-2]}\\
\ddots & \ddots & \ddots\\
&\ddots & \ddots & \ddots\\
&\Q_{n-\tau-k}^{[-(n-\tau-3)]} &\Q_{n-\tau-k-1}^{[-(n-\tau-3)]} & \Q_{n-\tau-k-2}^{[-(n-\tau-3)]}\\
&&\Q_{n-\tau-k}^{[-(n-\tau-2)]} &\Q_{n-\tau-k-1}^{[-(n-\tau-2)]} \\
&&&\Q_{n-\tau-k}^{[-(n-\tau-1)]}  \\
\end{pmatrix}\label{eq:factorization_system_submat}
\cdot
\begin{pmatrix}
\f_0\\
\f_1^{[-1]}\\
\vdots\\
\f_{k-1}^{[-(k-1)]}\\
\end{pmatrix}
=
\begin{pmatrix}
-\vec{q}_{0,0}\\
-\vec{q}_{0,1}^{[-1]}\\
\vdots\\
-\vec{q}_{0,n-\tau-1}^{[-(n-\tau-1)]}\\
\end{pmatrix}, \\[-3ex]
&\underbrace{\phantom{\hspace{51ex}}}_{} \hspace{2ex} \underbrace{\phantom{\hspace{11ex}}}_{} \hspace{5ex}\underbrace{\phantom{\hspace{14ex}}}_{}\nonumber\\[-2ex]
&\hspace{25ex}\Mat{Q}\hspace{24.5ex} \cdot \hspace{6ex}\vec{f}\hspace{5.5ex} = \hspace{7ex}\vec{q}_0\nonumber
\renewcommand{\arraystretch}{1}
\end{align}
where $\Q$ is an $((n-\tau)d_I)\times sk$ matrix and where we assume that $f^{(i)}_j = 0$ if $j \geq k^{(i)}$ and $q_{i,j} = 0$ when $j \geq n- \tau - k^{(i)}$, $\forall i\in\intervallincl{1}{s}$. 

\begin{lemma}[Complexity of the Root-Finding Step]
Let $Q^{(h)}(x,y_1,\dots,y_s)$, $\forall h \in \intervallincl{1}{d_I}$, be given, satisfying the interpolation constraints from Problem~\ref{prob:interpolation}.  
Then, the basis of the subspace, containing the coefficients of all tuples $(f^{(1)}(x), \dots, f^{(s)}(x))$ such that 
\begin{equation*}
F(x) = Q\left(x,f^{(1)}(x), \dots,f^{(s)}(x)\right) 
 = 0,
\end{equation*}
can be found recursively with complexity at most $\mathcal O(s^3 k^2)$ operations in $\Fqm$.
\end{lemma}
\begin{proof}
The complexity of calculating $q$-powers is negligible (compare e.g., \cite{Gadouleau_complexity}). 
The solution of \eqref{eq:factorization_system_submat} can be found by the following recursive procedure. 
First, solve the linear system of equations $\Q_0^{[0]} \cdot \f_0 = -\vec{q}_{0,0}$ of size $d_I \times s$ for $\f_0$ with complexity at most $\mathcal O(s^3)$ using Gaussian elimination.
Afterwards, calculate $\Q_1^{[-1]} \cdot \f_0$ with $sd_I \approx s^2$ multiplications over $\Fqm$ and solve the system $\Q_1^{[-1]} \cdot \f_0 + \Q_0^{[-1]} \cdot\f_1^{[-1]} = -\vec{q}_{0,1}^{[-1]} $ for $\f_1$ with complexity at most $\mathcal O(s^3)$ operations.
We continue this until we obtain all coefficients of $f^{(1)}(x), \dots, f^{(s)}(x)$, where for $\f_j$, we first have to calculate $(j-1)\cdot s \cdot d_I $ multiplications over $\Fqm$ and solve a $d_I \times s$ linear system of equations. Hence, the overall complexity for the root-finding step is upper bounded by
$\sum_{j=1}^{k} \left((j-1)\cdot s \cdot d_I + s^3\right) \leq \mathcal O(s^2k^2 + s^3k) \leq \mathcal O(s^3k^2)$
operations over $\Fqm$.
\qed\end{proof}

\section{Decoding Approaches}\label{sec:interpolation_decoders}
The decoding principle from the previous section can be used as a list decoding algorithm, returning all codewords of the interleaved Gabidulin code in rank distance at most $\tau$ from the received word, where $\tau$ satisfies \eqref{eq:decradius_het} (described in Subsection~\ref{subsec:intgab_list_decoding}), or as a probabilistic unique decoding algorithm (described in Section~\ref{subsec:intgab_uniquedecoding}).

\subsection{A List Decoding Approach}\label{subsec:intgab_list_decoding}
Our decoding approach for interleaved Gabidulin codes can be seen as a list decoding algorithm, consisting of solving two linear systems of equations. 
\begin{lemma}[Maximum List Size]
Let $\r^{(i)}$, $\forall i\in\intervallincl{1}{s}$, be given and let 
$\tau$ satisfy \eqref{eq:decradius_het}. 
Then, the list size $\ell_I$, i.e., the number of codewords from $\IntGab{s;n,k^{(1)},\dots,k^{(s)}}$ over $\Fqm$ in rank distance at most $\tau$ to $\r = (\r^{(1)T} \ \r^{(2)T} \ \dots \ \r^{(s)T} )^T$,
is upper bounded by:
\begin{equation*}
\ell_I  \defeq  \max_{\r \in \Fqm^{s \times n}}\left\{\big|\IntGab{s;n,k^{(1)},\dots,k^{(s)}} \cap \Ball{\tau}{\r}\big|\right\}\leq q^{m\left(\sum_{i=1}^{s}k^{(i)}- \min_i\{k^{(i)}\}\right)}.
\end{equation*}
\end{lemma}
\begin{proof}
The list size can be upper bounded by the maximum number of solutions of the root-finding step \eqref{eq:factorization_system_submat}.
There exists an integer $i\in \intervallincl{1}{s}$ such that $Q_i(x) \neq 0$, since $Q(x, y_1,\dots,y_s) \neq 0$.
Note that $Q_0(x)\neq 0$ and $Q_i(x) = 0$ , $\forall i \in \intervallincl{1}{s}$, is not possible since $(\Mooremat{n-\tau}{\vec{g}})^T$ is a full-rank matrix.

Hence, let $i \in \intervallincl{1}{s}$ be such that $Q_i(x) \neq 0$ and let $j$ be the smallest integer such that $q_{i,j} \neq 0$.
Consider the submatrix of $\Q$, which consists of the columns corresponding to the coefficients of $f^{(i)}(x)$. 
For some $h \in \intervallincl{1}{s}$, this submatrix contains at least one $k^{(i)} \times k^{(i)}$ lower triangular matrix 
with $q^{(h)[-j]}_{i,j},q^{(h)[-(j+1)]}_{i,j},\dots,$ $ q^{(h)[-(j+k^{(i)}-1)]}_{i,j}$ on the diagonal.
Therefore, $\rk(\Q)\geq \min_i\{k^{(i)}\}$ and the dimension of the solution space is at most $\sum_{i=1}^{s}k^{(i)}- \min_i\{k^{(i)}\}$. 
\qed\end{proof}
It is not clear whether the list size can really be that large. 
Finding the actual list of codewords out of the solution space of \eqref{eq:factorization_system_submat} further reduces the list size.


When $\ell_I>1$, the system of equations for the root-finding step \eqref{eq:factorization_system_submat} cannot have full rank. 
The following lemma estimates the average list size. For most parameters, this value is almost one (see Example~\ref{ex:failure_prob}).
The proof proceeds similar to McEliece's proof for the average list size in the Guruswami--Sudan algorithm \cite {McEliece_OntheaveragelistsizefortheGuruswami-Sudandecoder_2003}.
\begin{lemma}[Average List Size]\label{lem:average_listsize} 
Let   $\c^{(i)} = f^{(i)}(\vec{g})$, $\forall i \in \intervallincl{1}{s}$, where $\deg_q f^{(i)}(x)$ $ < k^{(i)}$ and let $\r^{(i)}  = \c^{(i)}+\e^{(i)}$. 
Let $\rk(\e^{(1)T} \ \e^{(2)T} \ \dots \ \e^{(s)T})  = t\leq \tau$ and let $\tau$ satisfy \eqref{eq:decradius_het}. 
Then, the average list size, i.e., the average number of codewords $(\c^{(1)T} \ \c^{(2)T} \ \dots \ \c^{(s)T} )^T \in \IntGab{s;n,k^{(1)},\dots,k^{(s)}}$ such that 
\begin{equation*}
\rk\left((\r^{(1)T} \ \r^{(2)T} \ \dots \ \r^{(s)T} )- (\c^{(1)T} \ \c^{(2)T} \ \dots \ \c^{(s)T} ) \right)\leq \tau,
\end{equation*}
is upper bounded by
\begin{equation*}
\overline {\ell_I} < 1 + 4\left(q^{m\sum_{i=1}^{s}k^{(i)}}-1\right)q^{(sm+n)\tau-\tau^2-smn}.
\end{equation*}
\end{lemma}
\begin{proof}
Let $R$ be a random variable, uniformly distributed over all matrices in $\Fqm^{s \times n}$ 
and let $\r$ be a realization of $R$, i.e., the $s$ elementary received words written as rows of a matrix.
Let $\c \in \IntGab{s;n,k^{(1)},\dots,k^{(s)}}$ be the fixed transmitted codeword. Then,
$P(\rk(\r -\c)\leq \tau) = P(\rk(\r)\leq \tau)$, which is 
the probability that a random $sm \times n$ matrix over $\Fq$ has rank at most $\tau$. 
Let $\IntGabstar{s;n,k^{(1)},\dots,k^{(s)}}$ be the code $\IntGab{s;n,k^{(1)},\dots,k^{(s)}}$ \emph{without} the transmitted codeword.

Let us further consider another random variable $X$, which depends on $R$:
\begin{equation*}
X(R)=\Big|\big\lbrace \IntGabNoInput^* \cap \Ball{\tau}{\r}\big\rbrace\Big|,
\end{equation*}
where $\r \in \Fqm^{s \times n}$. 
Denote by $\Indfunc(..)$ the indicator function, then the expectation of $X$ is given by:
\begin{align*}
E[X] &= \sum_{\r \in \Fqm^{s\times n}} P(R=\r) X(R)
=\sum_{\c \in \IntGabNoInput^*}\sum_{\r \in \Fqm^{s\times n}}  \Indfunc(\rk(\r-\c)\leq \tau) P(R=\r)\\
&= \!\sum_{\c \in\IntGabNoInput^* }\!\! E\big[\Indfunc(\rk(\r-\c)\leq \tau)\big]
=\!\sum_{\c \in\IntGabNoInput^* }\!\! P(\rk(\r-\c)\leq \tau)
=\!\sum_{\c \in \IntGabNoInput^*}\!\! P(\rk(\r)\leq \tau).
\end{align*}
Therefore,
\begin{align*}
E[X]&= \big|\IntGabNoInput^*\big| \cdot \frac{\big|\R \in \Fq^{sm \times n} : \rk(\R) \leq \tau\big|}{q^{smn}}\\
&< \left((q^m)^{\sum_{i=1}^{s}k^{(i)}}-1\right)\frac{4q^{(sm+n)\tau-\tau^2}}{q^{smn}}.
\end{align*}
The average list size is $\overline{\ell_I} = E[X]+1$ due to the transmitted codeword.
\qed\end{proof}
Unfortunately, it is not clear if it is possible that $\ell_I = 1$ and still, the system of equations for the root-finding step \eqref{eq:factorization_system_submat} does {not} have full rank. 
Thus, Lemma~\ref{lem:average_listsize} does not bound the probability that the rank of $\Q$ is not full;
this is done in Lemma~\ref{lem:intgab_failureproba_altern}.

Theorem~\ref{theo:list_decoding_int} summarizes the properties of our list decoding algorithm and Algorithm~\ref{algo:list_decoding} shows the steps of the decoder in pseudocode.

\begin{theorem}[List Decoding of Interleaved Gabidulin Codes]\label{theo:list_decoding_int}
Let the interleaved Gabidulin code $\IntGab{s;n,k^{(1)},\dots,k^{(s)}}$ over $\Fqm$ consist of $\c^{(i)} = f^{(i)}(\vec{g})$, where $\deg_q f^{(i)}(x)$ $ < k^{(i)}$, and let the elementary received words $\r^{(i)}$, $\forall i \in \intervallincl{1}{s}$, be given. 

Then, we can find a basis of the affine subspace,
containing all tuples of polynomials $(f^{(1)}(x), \dots, f^{(s)}(x))$, such that their evaluation at $\vec{g}$ is in rank distance 
\begin{equation*}
 \tau < \frac{sn - \sum_{i=1}^{s}k^{(i)}+s}{s+1}
\end{equation*}
from $(\r^{(1)T} \ \r^{(2)T} \ \dots \ \r^{(s)T})^T$ with overall complexity at most $\mathcal O(s^3n^2)$.
\end{theorem}
The complexity of finding the \emph{basis} of the list is quadratic in $n$, but the complexity for finding explicitly the whole list can be exponential in $n$. The dimension of the solution space of \eqref{eq:factorization_system_submat} is $(sk-\rk(\Mat{Q}))$ over $\Fqm$, which results in $q^{m(sk-\rk{\Mat{Q}})}$ possible solutions. If the rank of $\Mat{Q}$ is not full (i.e., less than $sk$), we have to examine all these solutions and check if they correspond to a valid codeword.
Therefore, the complexity of our list decoder depends on the rank of the matrix $\Q$, but not on the "real" list size. It is not clear if there is a connection between the real list size and the rank of $\Q$, going beyond the fact that if the real list size is greater than one, the rank of $\Q$ cannot be full.

Therefore, this is not a polynomial-time list decoder, although in most cases a unique solution can be found with quadratic time complexity.
\printalgo{
\caption{\newline$\List$  $\leftarrow$\textsc{ListDecodingInterleavedGabidulin}$\big(\r^{(1)}, \dots, \r^{(s)} \big)$}
\label{algo:list_decoding}
\DontPrintSemicolon
\SetAlgoVlined
\BlankLine
\LinesNumberedHidden
\SetKwInput{KwIn}{\underline{Input}}
\SetKwInput{KwOut}{\underline{Output}}
\SetKwInput{KwIni}{\underline{Initialize}}
\KwIn{$\r^{(i)} = \vecelements{r^{(i)}}\in \Fqm^n$ with $n\leq m$, $\forall i \in \intervallincl{1}{s}$}
\vspace{0.2ex}
\KwIni{$\List =\emptyset$}
\BlankLine
\textbf{Interpolation step:} \;
\vspace{0.2ex}
\ShowLn Define $\Mat{R}$ as in \eqref{eq:interpolation_matrix}\;
\ShowLn Solve $\Mat{R} \cdot \Mat{q} = \0$ for $\Mat{q}\in \Fqm^{n-\tau+\sum_{i=1}^{s}(n-\tau-k^{(i)}+1)}$ \nllabel{algoline:calcq}\;\vspace{0.2ex}
\ShowLn Define $Q(x,y_1,\dots,y_s) = Q_0(x)+Q_1(x)+\dots+Q_s(x)$ as in \eqref{eq:coeffs_multivpoly}, where $\Mat{q}$ is calculated in Line~\ref{algoline:calcq} \;
\BlankLine
\textbf{Root-finding step:} \;
\vspace{0.2ex}
\ShowLn Define $\Mat{Q}$ as in \eqref{eq:intgab_notation_Qij}, \eqref{eq:factorization_system_submat}\;\vspace{0.2ex}
\ShowLn Determine affine solution space of $\Mat{Q} \cdot \vec{f} = \vec{q}_0$ of dimension $(sk-\rk(\Mat{Q}))$\;\vspace{0.2ex}
\ShowLn Determine all vectors in this solution space and save them in set $\myset{F}$\;\vspace{0.2ex}
\ForEach{$\vec{f} = (\vec{f}^{(1)} \ \dots \ \vec{f}^{(s)})\in \myset{F}$}
{
\If{$\rk((\r^{(1)} \ \dots \ \r^{(s)})-(f^{(1)}(\vec{g}) \ \dots \ f^{(s)}(\vec{g}))) \leq \tau$}
{$\List \leftarrow \List \cup \vec{f}$}
}
\BlankLine
\KwOut{List of evaluation words $\List$}
\BlankLine}

\subsection{A Probabilistic Unique Decoding Approach}\label{subsec:intgab_uniquedecoding}
In this section, we apply our decoding approach to {probabilistic unique} decoding.
Since the list size might be greater than one
, there is not always a unique solution.
We accomplish the interpolation step as before and declare a {decoding failure} as soon as the rank of the root-finding matrix $\Q$ is not full (see \eqref{eq:factorization_system_submat}). 
We upper bound this probability and call it \emph{failure probability}.
The failure probability is actually the fraction of non-correctable error matrices.
We show a relation to the approaches from \cite {Loidreau_Overbeck_Interleaved_2006,SidBoss_InterlGabCodes_ISIT2010}. 
The upper bound as well as simulation results show that the failure probability is quite small. 
Therefore, we can use our decoder as probabilistic unique decoder which basically consists of solving two structured linear systems of equations and has overall complexity at most $\mathcal O(s^3n^2)$, where $s \ll n$ is usually a small fixed integer.

It is important to observe that we always set up the system of equations for the interpolation step (Problem~\ref{prob:interpolation})
with maximum possible $\tau$, but---in contrast to solving the systems of equations from \eqref{eq:intgab_lo_systemeq} and \eqref{eq:interleavedGabi_SidBoss}---we also find the unique solution (if it exists) if $t<\tau$ without decreasing the size of the matrix, since the rank of the matrix $\Mat{R}$ from \eqref{eq:interpolation_matrix} is not important. 

Recall the notations from \eqref{eq:intgab_notation_Qij}
and denote additionally the $d_I \times (s+1)$ matrix
\begin{equation}\label{eq:intgab_matrixQ0}
\overline{\Q}_0 \defeq \begin{pmatrix}
q^{(1)}_{0,0}&q^{(1)}_{1,0} & \dots & q^{(1)}_{s,0}\\
q^{(2)}_{0,0}&q^{(2)}_{1,0} & \dots & q^{(2)}_{s,0}\\
\vdots&\vdots &\ddots& \vdots \\
q^{(d_I)}_{0,0}& q^{(d_I)}_{1,0} & \dots & q^{(d_I)}_{s,0}\\
\end{pmatrix}.
\end{equation}
For any matrix $\A$ with entries in $\Fqm$ it holds that $\rk(\A^{[i]}) = \rk(\A)$ for any integer $i$. 
The matrix $\Q$ \eqref{eq:factorization_system_submat} contains a lower block triangular matrix, 
providing Lemma~\ref{lem:matfactsmall}.

\begin{lemma}[Rank of Root-Finding Matrix]\label{lem:matfactsmall}
Let $\Q$ be defined as in \eqref{eq:factorization_system_submat} and $\Q^{[0]}_0$ as in \eqref{eq:intgab_notation_Qij}.
If $\rk(\Q_0^{[0]}) = s$
, then $\rk(\Q)= sk$.
\end{lemma}
\begin{proof}
This holds since $\Q$ contains a lower block triangular matrix with $\mathbf{Q}_0^{[0]}$, $\dots$, $\mathbf{Q}_0^{[k-1]}$ on the diagonal of the first $k$ blocks and since $\rk(\Q_0^{[0]}) = \rk(\Q_0^{[i]})$.
\qed\end{proof}
The $d_I \times s$ matrix $\Q^{[0]}_0$ can have rank $s$ only if $d_I \geq s$, which is guaranteed for $t = \tau$ if (compare~\eqref{eq:intgab_dimsolspace_interpolation}):
\begin{equation}\label{eq:tau_restri_2}
d_I = \dim \ker (\R) \geq s(n-\tau+1)-\sum_{i=1}^{s}k^{(i)}-t \geq s \ \Longleftrightarrow \ t \leq \frac{sn - \sum_{i=1}^{s}k^{(i)}}{(s+1)}.
\end{equation}
This is equivalent to the decoding radius of joint decoding and slightly different to \eqref{eq:decradius_het}, which
is the maximum decoding radius when we consider our algorithm as a list decoder (see Section~\ref{subsec:intgab_list_decoding}).


Let us show a connection between the probability that $\Q$ does not have full rank and that the matrix $\R_R$ from \cite {Loidreau_Overbeck_Interleaved_2006}, see \eqref{eq:matrix_loidreauoverbeck}, 
does not have full rank.
\begin{lemma}[Connection Between Matrices of Different Approaches]\label{lem:connection_interpol_loidreau}
Let $\overline{\Q}_0$ be defined as in \eqref{eq:intgab_matrixQ0} and $\R_R$ as in \eqref{eq:matrix_loidreauoverbeck} for
$t = \tau = \lfloor(sn - \sum_{i=1}^{s}k^{(i)})/(s+1)\rfloor$. 
If $\rk(\overline{\Q}_0) < s$, then $\rk(\R_R) <n-1$. 
\end{lemma}
\begin{proof}
If $\rk(\overline{\Q}_0) < s$, then by linearly combining the $d_I \geq s$ dimensional basis of the solution space of the interpolation step, there exists a non-zero interpolation polynomial $Q(x,y_1,\dots,y_s)$, which fulfills Problem~\ref{prob:interpolation} and has the coefficients
$q_{0,0} = q_{1,0} = \dots = q_{s,0} = 0$. Since $Q(x,y_1,\dots,y_s) \neq 0$ (Lemma~\ref{lem:nonzero_interpolpoly}), the interpolation matrix without the first column of each submatrix (i.e., the columns corresponding to $q_{0,0},q_{1,0},\dots,q_{s,0}$), denoted by $\widetilde\R$, does not have full rank. 

Moreover $ \R_R^{[1]} = {\widetilde\R}^T$ and hence, 
\begin{equation*}
\rk(\R_R) = \rk(\widetilde\R) < \sum\limits_{i=0}^{s}\deg_q Q_i(x) = (s+1)(n-\tau)-\sum_{i=1}^{s}(k^{(i)} - 1).
\end{equation*}
For $\tau = \left\lfloor{(sn - \sum_{i=1}^{s}k^{(i)})}/{(s+1)}\right\rfloor$, this gives 
$\rk(\R_R) < n-1$. 
\qed\end{proof}
Combining the last two lemmas, we obtain the following theorem.
\begin{theorem}[Connection Between Failure Probabilities]\label{theo:our_failureprob}
Assume that $\r^{(i)}$,$\forall i \in \intervallincl{1}{s}$, consists of random elements uniformly distributed over $\Fqm$.
Let $\R_R$ be as in \eqref{eq:matrix_loidreauoverbeck} and $\S$ as in \eqref{eq:interleavedGabi_SidBoss} for
$t = \tau =\lfloor(sn - \sum_{i=1}^{s}k^{(i)})/(s+1)\rfloor$.
Then, for $k = \max_i\{k^{(i)}\}$:
\begin{equation}\label{eq:failure_proba_theo}
P\big(\rk(\Q)<sk\big) \leq P\big(\rk(\overline{\Q}_0) < s\big) \leq P\big(\rk(\R_R) < n-1\big). 
\end{equation}
Therefore, for $\tau \geq s$:
\begin{equation*}
P\big(\rk(\Q)<sk\big) \leq1 - \left(1-\frac{4}{q^m}\right)\left(1-q^{m(s-\tau)}\right)^s.
\end{equation*}
If $k^{(i)} = k$, $\forall i\in \intervallincl{1}{s}$, additionally $P\big(\rk(\Q)<sk\big) \leq P\big(\rk(\S) < \tau\big)$ holds.
\end{theorem}
\begin{proof}
Since $\tau = \lfloor(sn - \sum_{i=1}^{s}k^{(i)})/(s+1)\rfloor$, we obtain $d_I = s$ and $\rk(\overline{\Q}_0) = \rk(\Q_0^{[0]})$. 
The first inequality of \eqref{eq:failure_proba_theo} follows from Lemma~\ref{lem:matfactsmall} and the second from Lemma~\ref{lem:connection_interpol_loidreau}.
Hence, we can bound $P\big(\rk(\Q)<sk)$ by the failure probability from \cite {Loidreau_Overbeck_Interleaved_2006}. 
Due to Lemma~\ref{lem:intgab_rel_LO_SB}, the failure probability from \cite {Loidreau_Overbeck_Interleaved_2006} is the same as the one from \cite {SidBoss_InterlGabCodes_ISIT2010} for $k^{(i)} = k$, $\forall i\in\intervallincl{1}{s}$.
\qed\end{proof}
The assumption of random received vectors and the restriction $\tau \geq s$ follow from \cite[Theorem~3.11]{Overbeck_Diss_InterleveadGab}. 
We conjecture that $\tau \geq s$ is only a technical restriction and that the results hold equivalently for $\tau < s$.

Alternatively, we can bound the failure probability as follows. Assume,
the matrix $\Q_0^{[0]}$ consists of random values over $\Fqm$. This assumption seems to be reasonable, since in \cite {Loidreau_Overbeck_Interleaved_2006} and \cite {SidBoss_InterlGabCodes_ISIT2010} it is assumed that 
$\r^{(1)},\r^{(2)},  \dots, \r^{(s)}$ are random vectors in $\Fqm^n$. In our approach, the values of $\Q_0^{[0]}$ are obtained from a linear system of equations, where each $q_{i,0}$ is multiplied with the coefficients of a different $\r^{(i)}$.
\begin{lemma}[Alternative Calculation of Failure Probability]\label{lem:intgab_failureproba_altern}
Let\\ $\rk\big(\e^{(1)T} \ \e^{(2)T} \ \dots \ \e^{(s)T} \big)  = t\leq \tau$, where $\tau =\lfloor(sn - \sum_{i=1}^{s}k^{(i)})/(s+1)\rfloor$, let $k = \max_i\{k^{(i)}\}$, let $\Mat{Q}$ be defined as in \eqref{eq:factorization_system_submat} and let $q^{(j)}_{1,0}, q^{(j)}_{2,0},\dots, q^{(j)}_{s,0}$ for $j = 1,\dots, d_I$ be random elements uniformly distributed over $\Fqm$. 
Then, 
\begin{equation*}
P\big(\rk(\Q) < sk\big) \leq \frac{4}{q^{(m(d_I+1-s))}}= 4 q^{-m\left(s(n-\tau)-\sum_{i=1}^{s}k^{(i)}-t+1\right)}.
\end{equation*}
\end{lemma}
\begin{proof}
Due to $d_I \geq s$ and Lemma~\ref{lem:matfactsmall}, if $\rk(\Q^{[0]}_0) = s$, then $\rk(\Q)=sk$. Hence,
$P\left(\rk(\Q) < sk\right) \leq P(\rk(\Q^{[0]}_0) < s)$. When $q^{(j)}_{1,0},  \dots, q^{(j)}_{s,0}$, $\forall j \in \intervallincl{1}{d_I}$, are random elements from $\Fqm$, we can bound $P(\rk(\Q^{[0]}_0) < s)$ by the probability that a random $(d_I \times s)$-matrix over $\Fqm$ has rank less than $s$:
 \begin{align*}
\hspace{5ex} P\big(\rk(\Q) < sk\big) &\leq P\big(\rk(\Q^{[0]}_0)< s\big) \\
& \leq \frac{\sum\limits_{j=0}^{s-1}\prod\limits_{h=0}^{j-1}\frac{q^{d_I}-q^h}{q^j-q^h} \prod\limits_{i=0}^{j-1} (q^s-q^i)}{q^{m s d_I }} 
 < \frac{4q^{m\left((d_I+s)(s-1)-(s-1)^2\right)}}{q^{m s d_I }}\\
  &= \frac{4}{q^{m(d_I-s+1)}} = 4 q^{-m\left(s(n-\tau)-\sum_{i=1}^{s}k^{(i)}-t+1\right)}.
 \end{align*}
\qed\end{proof}
Lemma~\ref{lem:intgab_failureproba_altern} does not have the technical restriction $\tau \geq s$ as Theorem~\ref{theo:our_failureprob} and the bounds from \cite{Loidreau_Overbeck_Interleaved_2006,SidBoss_InterlGabCodes_ISIT2010}. 
Theorem~\ref{thm:unique_decoding} summarizes our results, Algorithm~\ref{algo:unique_decoding} summarizes the steps of our decoder in pseudocode and Example~\ref{ex:failure_prob} illustrates the failure probability. 

\begin{theorem}[Unique Decoding of Interleaved Gabidulin Codes]\label{thm:unique_decoding}
Let the interleaved Gabidulin code $\IntGab{s;n,k^{(1)},\dots,k^{(s)}}$ over $\Fqm$ consist of the elementary codewords $\c^{(i)} = f^{(i)}(\vec{g})$, where $\deg_q f^{(i)}(x)$ $ < k^{(i)}$, $\forall i \in \intervallincl{1}{s}$, and let the given elementary received words $\r^{(i)}$, $\forall i \in \intervallincl{1}{s}$, consist of random elements uniformly distributed over $\Fqm$.
Then, with probability at least
\begin{equation*}
1- 4 q^{-m\left(s(n-\tau)-\sum_{i=1}^{s}k^{(i)}-t+1\right)},
\end{equation*} 
we can find a unique solution $f^{(1)}(x), \dots, f^{(s)}(x)$ 
such that its evaluation at $\vec{g}$ is in rank distance 
\begin{equation*}
t \leq \tau = \Big\lfloor\frac{sn   - \sum_{i=1}^{s}k^{(i)}}{s+1}\Big\rfloor
\end{equation*}
to $(\r^{(1)T} \ \r^{(2)T} \ \dots \ \r^{(s)T})^T$ with overall complexity at most $\mathcal O(s^3n^2)$.
\end{theorem}

\printalgoWidth{
\caption{\newline$f^{(1)}(x), \dots, f^{(s)}(x)$ or "decoding failure"  $\leftarrow$\textsc{UniqueDecIntGab}$\big(\r^{(1)}, \dots, \r^{(s)} \big)$}
\label{algo:unique_decoding}
\DontPrintSemicolon
\SetAlgoVlined
\BlankLine
\LinesNumberedHidden
\SetKwInput{KwIn}{\underline{Input}}
\SetKwInput{KwOut}{\underline{Output}}
\SetKwInput{KwIni}{\underline{Initialize}}
\KwIn{$\r^{(i)} = \vecelements{r^{(i)}}\in \Fqm^n$ with $n\leq m$, $\forall i \in \intervallincl{1}{s}$}
\BlankLine
\textbf{Interpolation step:} \;
\vspace{0.2ex}
\ShowLn Define $\Mat{R}$ as in \eqref{eq:interpolation_matrix}\;
\ShowLn Solve $\Mat{R} \cdot \Mat{q} = \0$ for $\Mat{q}\in \Fqm^{n-\tau+\sum_{i=1}^{s}(n-\tau-k^{(i)}+1)}$\nllabel{algoline:calcq_2} \;\vspace{0.2ex}
\ShowLn Define $Q(x,y_1,\dots,y_s) = Q_0(x)+Q_1(x)+\dots+Q_s(x)$ as in \eqref{eq:coeffs_multivpoly} where $\vec{q}$ is calculated in Line~\ref{algoline:calcq_2}\;
\BlankLine
\textbf{Root-finding step:} \;
\vspace{0.2ex}
\ShowLn Define $\Mat{Q}$ as in \eqref{eq:intgab_notation_Qij}, \eqref{eq:factorization_system_submat}\;\vspace{0.2ex}
\ShowLn \If{$\rk(\Mat{Q)}=sk$}
{\BlankLine
Solve $\Mat{Q} \cdot \vec{f} = \vec{q}_0$ for $\vec{f}$\;
Define $f^{(1)}(x), \dots, f^{(s)}(x)$ from $\vec{f}$\;
\KwOut{$f^{(1)}(x), \dots, f^{(s)}(x)$}
}
\Else{\KwOut{"decoding failure"}}
\BlankLine}{1}

\begin{example_bold}[Failure Probabilities]\label{ex:failure_prob}
Consider $\IntGabNoInput[s=2;n=7,k^{(1)} = k^{(2)}=2]$ code over $\mathbb{F}_{2^7}$.
The maximum decoding radius for unique as well as for list decoding according to \eqref{eq:decradius_het} and \eqref{eq:tau_restri_2} is $\tau = 3$ whereas a BMD decoder guarantees to correct all errors of rank at most $\Uniquecorrcap=2$.

In order to estimate the failure probability,
we first simulated $10^7$ random error matrices $(\e^{(1)T} \ \e^{(2)T} \ \dots \ \e^{(s)T} )^T\in\Fqm^{s \times n}$, uniformly distributed over all matrices of rank $t = \tau = 3$. 
The following simulated probabilities occurred:
\begin{equation*}
P\big(\rk(\Q) < sk\big)  = P\big(\rk(\S) < \tau\big) = P\big(\rk(\R_R) < n-1\big) = 6.12\cdot 10^{-5}.
\end{equation*}
As a comparison, the average list size calculated with Lemma~\ref{lem:average_listsize} is $\overline{\ell_I} < 1 + 6.104 \cdot 10^{-5}$,
the upper bound from Theorem~\ref{theo:our_failureprob} (and therefore the upper bound from \eqref{eq:intgab_failure_lo}, \cite {Loidreau_Overbeck_Interleaved_2006}) gives 
\begin{equation*}
P\big(\rk(\Q) < sk\big) \leq P\big(\rk(\R_R) < n-1\big) \leq 0.04632, 
\end{equation*}
and the bound from Lemma~\ref{lem:intgab_failureproba_altern} gives
\begin{equation*}
P\big(\rk(\Q) < sk\big) \leq 4 q^{-m\left(s(n-k-\tau)-\tau+1\right)} = 2.44 \cdot 10^{-4}.
\end{equation*}

Second, in order to estimate the performance compared to BMD decoding, we simulated $10^7$ transmissions over a $q$-ary symmetric rank channel, which is defined in analogy to 
the usual $q$-ary symmetric channel such that
\begin{equation*}
P\left(\rk (\e^{(1)T} \ \e^{(2)T} \ \dots \ \e^{(s)T} ) = t \right) = \binom{n}{t} p_{qsc}^t(1-p_{qsc})^{n-t}.
\end{equation*}
Fig.~\ref{fig:sim_interleavedgabidulin} shows the block error probability of the transmission of the $\IntGabNoInput[s=2;n=7,k^{(1)} = k^{(2)}=2]$ code over such a symmetric rank channel.
The result is dominated by the probability of $t > \tau$, where all four shown decoders always fail.
The list decoder only fails when $t > \tau$.
However, Fig.~\ref{fig:sim_interleavedgabidulin} shows that the failure probability for $t \leq \tau$ is almost negligible compared to the probability that $t > \tau$.
\end{example_bold}
\begin{figure}[htb]
\centering
\includegraphics[scale=1]{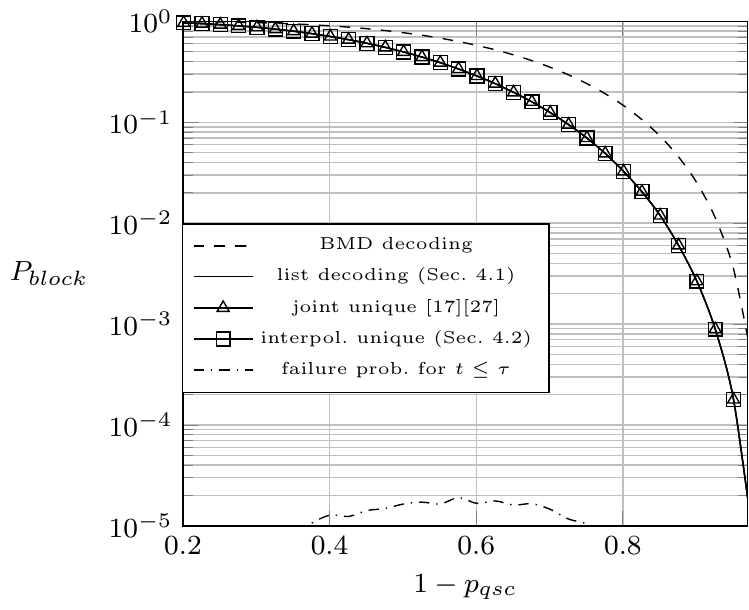}
\caption{Simulation results for a $\IntGabNoInput[s=2;n=7,k^{(1)} = 2,k^{(2)}=2]$ code over $\mathbb{F}_{2^7}$ with $10^7$ transmissions over a $q$-ary symmetric rank channel.}
\label{fig:sim_interleavedgabidulin}
\end{figure}

From the simulation results, we conjecture that 
$\rk(\Q)<sk$ \emph{if and only if} $\rk(\S) <\tau$ and $\rk(\R_R)<n-1$, i.e., that Lemma~\ref{lem:connection_interpol_loidreau} holds in \emph{both} directions.

Compared to the unique decoding approaches from \cite{Loidreau_Overbeck_Interleaved_2006,SidBoss_InterlGabCodes_ISIT2010}, our unique decoder achieves the same asymptotic time complexity and at most the same failure probability. The main advantage of our approach is that it can directly be used as a list decoder if an application can take advantage of that.




\section{Error-Erasure Decoding}\label{subsec:intgab_errorerasure}
Applications like random linear network coding provide additional side information about the occurred error, see e.g. \cite{silva_rank_metric_approach}.
Such information can be used to declare erasures and thus, to increase the decoding performance. 
In comparison to classical erasure decoders in Hamming metric, we distinguish two types of erasures in rank metric: row erasures and column erasures.
This section provides a generalization of our approach to interpolation-based error-erasure decoding of {interleaved} Gabidulin codes over $\Fqm$ with $n=m$. 
We consider the most general form of row and column erasures as in \cite{silva_rank_metric_approach,GabidulinPilipchuck_ErrorErasureRankCodes_2008} and show how the additional information can be incorporated into our decoding algorithm from the previous sections.

Notice that in this section, we consider only $n=m$. 
On the one hand, this simplifies the notations, but on the other hand, Lemma~\ref{lem:silva_qreverse_matrix} only holds for $n=m$. 

We show that the presented error-erasure list decoding approach is able
to reconstruct all codewords of an $\IntGab{s;n,k^{(1)},\dots,k^{(s)}}$ code over $\Fqm$ for $n=m$ with asymptotic complexity $\OCompl{n^2}$ operations over $\Fqm$ in distance at most 
\begin{equation*}
\tau < \frac{sn-\sum_{i=1}^{s}(k^{(i)}+\numbRowErasures^{(i)}+\numbColErasures)+s}{s+1},
\end{equation*}
from the received word, where $\numbRowErasures^{(i)}$ denotes the rank of the row erasures and $\numbColErasures$ the rank of the column erasures (see also following description).

\subsection{Row and Column Erasures and the Generalized Key Equation}
Let $\NormbasisOrdered = (\Normelement^{[0]} \ \Normelement^{[1]} \ \dots \ \Normelement^{[n-1]})$ and $\NormbasisDualOrdered = (\Normdualelement^{[0]} \ \Normdualelement^{[1]} \ \dots \ \Normdualelement^{[n-1]})$ denote (ordered) normal bases of $\Fqm$ over $\Fq$, which are dual to each other (see, e.g., \cite{Menezes2010Applications}).
Consider an $\IntGab{s;n,k^{(1)},\dots,k^{(s)}}$ code over $\Fqm$ of length $n=m$, defined by $\vec{g} = (\Normdualelement^{[0]} \ \Normdualelement^{[1]} \ \dots \ \Normdualelement^{[n-1]})$ as in Definition~\ref{def:int_gab}. The parity-check matrices of the elementary $\Gab{n,k^{(i)}}$ codes are $\Mooremat{n-k^{(i)}}{\vec{h}^{(i)}} = \Mooremat{n-k^{(i)}}{(\Normelement^{[k]} \ \Normelement^{[k^{(i)}+1]} \ \dots \ \Normelement^{[k^{(i)}+n-1]})}$, $\forall i \in \intervallincl{1}{s}$.
We assume that side information of the channel is given in form of:
\begin{itemize}
\item[$\bullet$] $\numbRowErasures^{(i)}$, $\forall i \in \intervallincl{1}{s}$, {row erasures} (in \cite{silva_rank_metric_approach} called "deviations") and
\item[$\bullet$] $\numbColErasures$ {column erasures} (in \cite{silva_rank_metric_approach} called "erasures"),
\end{itemize}  
such that the interleaved error matrix can be rewritten by:
\begin{equation}\label{eq:decomp_errorerasures_vector_intgab}
\begin{pmatrix}
\e^{(1)}\\
\e^{(2)}\\
\vdots\\
\e^{(s)}\\
\end{pmatrix}
= \begin{pmatrix}
\a^{(1,R)}\cdot \B^{(1,R)}\\
\a^{(2,R)}\cdot \B^{(2,R)}\\
\vdots\\
\a^{(s,R)}\cdot \B^{(s,R)}\\
\end{pmatrix}+ 
\begin{pmatrix}
\a^{(1,C)}\\
\a^{(2,C)}\\
\vdots\\
\a^{(s,C)}\\
\end{pmatrix}
\cdot \B^{(C)} + 
\begin{pmatrix}
\a^{(1,E)}\\
\a^{(2,E)}\\
\vdots\\
\a^{(s,E)}\\
\end{pmatrix} 
\cdot\B^{(E)} \in \Fqm^{s \times n},
\end{equation}
where $\a^{(i,R)} \in \Fqm^{\numbRowErasures^{(i)}}$, $\B^{(i,R)} \in \Fq^{\numbRowErasures^{(i)} \times n}$, 
$\a^{(i,C)} \in \Fqm^{\numbColErasures}$, $\B^{(C)} \in \Fq^{\numbColErasures\times n}$,
$\a^{(i,E)} \in \Fqm^{t}$, $\B^{(E)} \in \Fq^{t \times n}$ for all $i\in\intervallincl{1}{s}$,
and $\a^{(1,R)}, \a^{(2,R)},\dots, \a^{(s,R)}$ and $\B^{(C)}$ are known on the receiver side.

This decomposition is also shown in Fig.~\ref{fig:error_erasure} for one $\e^{(i)}$, where $\e^{(i)}$ as well as $\a^{(i,R)}$, $\a^{(i,C)}$ and $\a^{(i,E)}$ are represented by their corresponding matrices over $\Fq$.
\begin{figure}[htb, scale = 1]
\centering
\includegraphics[scale=1]{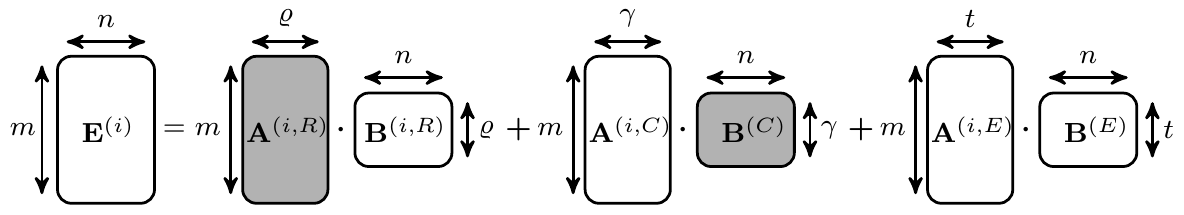}
\caption{Illustration of interleaved row erasures, column erasures and (full) errors in rank metric. The known matrices (given by the channel) are filled with gray.}
\label{fig:error_erasure}
\end{figure}

Lemma~\ref{lem:interleaved_mod_error} shows later why the $\a^{(i,R)}$ and $\B^{(i,R)}$ can be different whereas $\B^{(C)}$ has to be common for all $i\in \intervallincl{1}{s}$.
This model of errors and erasures is slightly more general than the one in \cite [Eq.~(19)]{LiSidorenkoChen-TransformDomainDecodingGabidulin}, where all $\a^{(i,R)}$ are equal.

%
%

Based on the known matrix $\B^{(C)}$, we can calculate the following basis of the row space of the column erasures prior to the decoding process:
\begin{equation}\label{eq:calc_di_errorerasure_}
d^{(C)}_i= \sum_{j=0}^{n-1}B^{(C)}_{i,j}g_j^\perp = \sum_{j=0}^{n-1}B^{(C)}_{i,j}\Normelement^{[j]},\quad \forall i\in \intervallexcl{0}{\numbColErasures}.
\end{equation}
Further, we define the linearized polynomials $\Gamma^{(C)}(x)$, $\Lambda^{(i,R)}(x)$ and $\Lambda^{(i,E)}$, $\forall i \in \intervallincl{1}{s}$, as linearized polynomials of smallest $q$-degree such that:
\begin{align}
\Gamma^{(C)}\big(d^{(C)}_j\big) &=0, \quad \forall j\in \intervallexcl{0}{\numbColErasures},\nonumber\\
\Lambda^{(i,R)}\big(a^{(i,R)}_j\big) &= 0, \quad \forall j\in \intervallexcl{0}{\numbRowErasures^{(i)}},\ i \in \intervallincl{1}{s},\nonumber\\
\Lambda^{(i,E)}\big(\Lambda^{(i,R)}(a^{(i,E)}_j)\big) &= 0, \quad\forall  j \in \intervallexcl{0}{t},\ i \in \intervallincl{1}{s}.\label{eq:defi_errorloc_span_errdev_intgab}
\end{align}
Therefore, $\Gamma^{(C)}(x)$ and $\Lambda^{(i,R)}(x)$, $\forall i \in \intervallincl{1}{s}$, can be calculated in the beginning of the decoding process since $\a^{(i,R)}$, $\forall i \in \intervallincl{1}{s}$, and $\B^{(C)}$ are known.

In the following, let $\qreciproc{p}(x) = \sum_{j=0}^{m-1}\qreciproc{p}_jx^{[j]}$ denote the \emph{full $q$-reverse linearized polynomial} of $p(x) \in \Linpolyring$, defined by the coefficients $\qreciproc{p}_j = p_{-j \!\mod\! m}^{[j]}$, $\forall j \in \intervallincl{0}{m}$, as in \cite{SilvaKschischang-FastEncodingDecodingGabidulin-2009,Silva_PhD_ErrorControlNetworkCoding}.
The following lemma shows that for $n=m$, the full $q$-reverse is closely related to the transpose of the associated evaluated matrix of $p(x)$.

\begin{lemma}[Evaluated Matrix of $q$-Reverse {\cite[Lemma~6.3]{Silva_PhD_ErrorControlNetworkCoding}}]\label{lem:silva_qreverse_matrix}
Let $p(x)\in \Linpolyring$, $\deg_q p(x)<m$, and its full $q$-reverse $\qreciproc{p}(x)$ with $\qreciproc{p}_i = p_{-i \mod m}^{[i]}$, for $i \in \intervallexcl{0}{m}$, be given. Let $\mathcal A = \setelements{\alpha}{m}$ and  $\mathcal B = \setelements{\beta}{m}$ be bases of $\,\Fqm$ over $\Fq$ and let $\mathcal A^\perp = \setelements{\alpha^\perp}{m}$ and $\mathcal B^\perp = \setelements{\beta^\perp}{m}$ denote their dual bases.
Let
\begin{equation*}
\vecevalm{p}{\alpha} = \vecelementsm{\beta} \cdot \Mat{P}, 
\end{equation*}
where $\Mat{P}\in \Fq^{m \times m}$. Then,
\begin{equation*}
\vecevalm{\qreciproc{p}}{\beta^\perp} = \vecelementsm{\alpha^\perp} \cdot \Mat{P}^T.
\end{equation*}
\end{lemma}
Based on this lemma, we can establish a (transformed) key equation, incorporating errors and row/column erasures, which is important for the proof of Lemma~\ref{lem:interleaved_mod_error}.

\begin{theorem}[Transformed Key Equation {\cite[App.~A.2]{Wachterzeh_PhD_2013}}]\label{theo:gen_trans_key_equation}
Let $\r^{(i)} = \c^{(i)} +\e^{(i)}$, with $\c^{(i)} \in \Gab{n,k^{(i)}}$, $\forall i \in \intervallincl{1}{s}$, over $\Fqm$ with $n=m$, 
be the given elementary received words and let
$\qtrafo{r}^{(i)}(x) = f^{(i)}(x)+\qtrafo{e}^{(i)}(x)$ be their linearized interpolation polynomial as in \eqref{eq:received_lagrange}.
Let $\Gamma^{(C)}(x)$, $\Lambda^{(i,R)}(x)$ and $\Lambda^{(i,E)}(x)$ be defined as in \eqref{eq:defi_errorloc_span_errdev_intgab}.

Then, these polynomials satisfy the following transformed key equation:
\begin{equation}\label{eq:transformed_keyequation_gen}
\Lambda^{(i,E)}\Big(\Lambda^{(i,R)} \big(\qtrafo{e}^{(i)}(\qreciproc{\Gamma^{(C)}}(x))\big)\Big) 
\equiv 0 \mod (x^{[m]}-x),\quad \forall i \in \intervallincl{1}{s}.
\end{equation}
\end{theorem}

\subsection{Error-Erasure Decoding of Interleaved Gabidulin Codes}

Let $\qtrafo{r}^{(i)}(x)$ denote the linearized interpolation polynomial of $r^{(i)}(x)$, $ \forall i\in \intervallincl{1}{s}$, calculated as in \eqref{eq:received_lagrange}, \eqref{eq:linearized_Lagrange_basis_poly}
and define $s$ modified transformed received words by:
\begin{equation}\label{eq:mod_rec_word}
\qtrafo{y}^{(i)}(x) \defeq \Lambda^{(i,R)} \big(\qtrafo{r}^{(i)}(\qreciproc{\Gamma^{(C)}}(x^{[\numbColErasures]}))\big) \mod (x^{[m]}-x), \quad \forall i \in \intervallincl{1}{s},
\end{equation}
where $\qreciproc{\Gamma^{(C)}}(x)$ is the full $q$-reverse of $\Gamma^{(C)}(x)$,  defined by $\qreciproc{\Gamma}_i = \Gamma_{-i \!\mod\! m}^{[i]}$, $\forall i \in \intervallincl{0}{m}$, see also \cite{silva_rank_metric_approach}.
These modified received words can immediately be calculated since all polynomials on the RHS of \eqref{eq:mod_rec_word} are known from the channel. Further,
\begin{equation}\label{eq:errorerasure_modifiedreceivedword}
\qtrafo{y}^{(i)}(x) =\underbrace{\Lambda^{(i,R)} \big(f^{(i)}(\qreciproc{\Gamma^{(C)}}(x^{[\numbColErasures]}))\big)}_{\deg_q < k^{(i)}+\numbRowErasures^{(i)}+\numbColErasures}+\,\Lambda^{(i,R)} \big(\qtrafo{e}^{(i)}(\qreciproc{\Gamma}^{(C)}(x^{[\numbColErasures]}))\big) \mod (x^{[m]}-x),
\end{equation}
where $f^{(i)}(x)$ with $\deg_q f^{(i)}(x) < k^{(i)}$ is the evaluation polynomial of the $i$-th elementary codeword such that $\c^{(i)} =f^{(i)}(\vec{g}) \in \Gab{n,k^{(i)}}$. 
The idea is to pass the evaluation of the modified transformed received words $\qtrafo{y}^{(i)}(x)$, $\forall i \in \intervallincl{1}{s}$, from \eqref{eq:errorerasure_modifiedreceivedword}---instead of the evaluation of $\qtrafo{r}^{(i)}(x)$---to our interpolation-based decoder. 

The polynomial $\Lambda^{(i,R)} \big(f^{(i)}(\qreciproc{\Gamma^{(C)}}(x^{[\numbColErasures]}))\big)$ on the RHS of \eqref{eq:errorerasure_modifiedreceivedword} has $q$-degree less than $k^{(i)}+\numbRowErasures^{(i)}+\numbColErasures$ and is the evaluation polynomial of a $\Gab{n,k^{(i)}+\numbRowErasures^{(i)}+\numbColErasures}$ codeword. 
If we arrange these polynomials for all $i \in \intervallincl{1}{s}$ vertically, we obtain the evaluation polynomial of an $\IntGab{s;n,k^{(1)}+\numbRowErasures^{(1)}+\numbColErasures,\dots,k^{(s)}+\numbRowErasures^{(s)}+\numbColErasures}$ code.

We call
$\Lambda^{(R)} \big(\qtrafo{e}(\qreciproc{\Gamma^{(C)}}(x^{[\numbColErasures]}))\big)$ {modified transformed error} in the following and show in Lemma~\ref{lem:interleaved_mod_error} that its evaluation has rank at most $t$.

\begin{lemma}[Rank of Modified Interleaved Error]\label{lem:interleaved_mod_error}
%
Let $n = m$ and let $\e^{(i, RC)} = \Lambda^{(i,R)} \big(\qtrafo{e}^{(i)}(\qreciproc{\Gamma^{(C)}}(\vec{g}^{[\numbColErasures]}))\big) \in \Fqm^n$, $\forall i \in \intervallincl{1}{s}$. 
Further, let $\e^{(i,E)} = \a^{(i,E)} \cdot \B^{(E)}$, $\forall i\in \intervallincl{1}{s}$ as in \eqref{eq:decomp_errorerasures_vector_intgab} with $\rk\big(\e^{(1,E)T} \ \e^{(2,E)T} \ \dots \ \e^{(s,E)T} \big)  = t$.
Then, 
\begin{equation*}
\rk
\begin{pmatrix}
\e^{(1,RC)}\\
\e^{(2,RC)}\\
\vdots\\
\e^{(s,RC)}\\
\end{pmatrix}
\leq 
\rk
\begin{pmatrix}
\e^{(1,E)}\\
\e^{(2,E)}\\
\vdots\\
\e^{(s,E)}\\
\end{pmatrix} = t.
\end{equation*}
\end{lemma}
\begin{proof}
Since $\Lambda^{(i,R)} \big((\qtrafo{e}^{(i,R)}(x)+\qtrafo{e}^{(i,C)}(x))\otimes \qreciproc{\Gamma^{(C)}}(x^{[\numbColErasures]}))\big) \equiv 0 \mod (x^{[m]}-x)$, $\forall i \in \intervallincl{1}{s}$, see proof of Theorem~\ref{theo:gen_trans_key_equation}, we obtain
\begin{equation*}
\Lambda^{(i,R)} \big(\qtrafo{e}^{(i)}(\qreciproc{\Gamma^{(C)}}(x^{[\numbColErasures]}))\big) 
\equiv\Lambda^{(i,R)} \big(\qtrafo{e}^{(i,E)}(\qreciproc{\Gamma^{(C)}}(x^{[\numbColErasures]}))\big) \mod (x^{[m]}-x), \forall i \in \intervallincl{1}{s}.
\end{equation*}
Let $\G = \big(G_{l,j}\big)^{l \in \intervallexcl{0}{m}}_{j \in \intervallexcl{0}{m}} \in \Fq^{m \times m}$ be such that $\qreciproc{\Gamma^{(C)}}(g_j^{[\numbColErasures]}) = \sum_{l=0}^{m-1}G_{l,j} g_l$ and thus, $\forall j\in\intervallexcl{0}{m}$ and $\forall i \in \intervallincl{1}{s}$:
\begin{align*}
e^{(i,RC)}_j &= \Lambda^{(i,R)} \big(\qtrafo{e}^{(i)}(\qreciproc{\Gamma^{(C)}}(g_j^{[\numbColErasures]}))\big)
= \sum_{l=0}^{m-1}G_{l,j}\Lambda^{(i,R)} \big(\qtrafo{e}^{(i,E)}(g_l)\big).
\end{align*}
Hence,
\begin{equation*}
\begin{pmatrix}
\e^{(1,RC)}\\
\e^{(2,RC)}\\
\vdots\\
\e^{(s,RC)}\\
\end{pmatrix} \!=\! 
\begin{pmatrix}
\Lambda^{(1,R)} \big(\qtrafo{e}^{(1,E)}(g_0)\big) & \Lambda^{(1,R)} \big(\qtrafo{e}^{(1,E)}(g_1)\big) & \dots & \Lambda^{(1,R)} \big(\qtrafo{e}^{(1,E)}(g_{m-1})\big)\\
\Lambda^{(2,R)} \big(\qtrafo{e}^{(2,E)}(g_0)\big) & \Lambda^{(2,R)} \big(\qtrafo{e}^{(2,E)}(g_1)\big) & \dots & \Lambda^{(2,R)} \big(\qtrafo{e}^{(2,E)}(g_{m-1})\big)\\
\vdots &\vdots&\ddots& \vdots\\
\Lambda^{(s,R)} \big(\qtrafo{e}^{(s,E)}(g_0)\big) & \Lambda^{(s,R)} \big(\qtrafo{e}^{(s,E)}(g_1)\big) & \dots & \Lambda^{(s,R)} \big(\qtrafo{e}^{(s,E)}(g_{m-1})\big)
\end{pmatrix}
\cdot \G.
\end{equation*}

Due to Lemma~\ref{lem:eval_row_space} in the appendix, $\big(\Lambda^{(i,R)} (\qtrafo{e}^{(i,E)}(g_0)) \ \Lambda^{(i,R)} (\qtrafo{e}^{(i,E)}(g_1)) \ \dots $\\ $ \ \Lambda^{(i,R)} (\qtrafo{e}^{(i,E)}(g_{m-1}))\big)$ lies in the same row space as $\e^{(i,E)} = \big(\qtrafo{e}^{(i,E)}(g_0) \ \qtrafo{e}^{(i,E)}(g_1) \ \dots $ $ \ \qtrafo{e}^{(i,E)}(g_{m-1})\big)$, $\forall i \in \intervallincl{1}{s}$, and hence, has rank at most $t^{(i)}$. The multiplication with $\G$ does not increase the rank and the statement follows.
\qed\end{proof}
Lemma~\ref{lem:interleaved_mod_error} requires that $\Gamma^{(C)}(x)$ is common for all $i\in\intervallincl{1}{s}$, whereas the $\Lambda^{(i,R)}(x)$ can be different.
This clarifies why $\B^{(C)}$ has to be independent of $i$.

Therefore, error-erasure decoding of an $\IntGab{s;n,k^{(1)},\dots,k^{(s)}}$ is reduced to errors-only decoding of an $\IntGab{s;n,k^{(1)}+\numbRowErasures^{(1)}+\numbColErasures,\dots,k^{(s)}+\numbRowErasures^{(s)}+\numbColErasures}$ code. In principle, any error decoding algorithm for interleaved Gabidulin codes can now be applied, e.g. our interpolation-based principle.
Hence, we use $\qtrafo{y}^{(i)}(\vec{g})$, $\forall i\in\intervallincl{1}{s}$, as the input of interpolation-based decoding and treat $\qtrafo{y}^{(i)}(\vec{g})$, $\forall i \in \intervallincl{1}{s}$, in the same way as a codeword of an interleaved Gabidulin code of elementary dimensions $k^{(i)}+\numbRowErasures^{(i)}+\numbColErasures$, which is corrupted by an error of overall rank $t$ and by no erasures.

In order to apply the interpolation-based decoding strategy, we generalize Problem~\ref{prob:interpolation} as follows.
\begin{problem_bold}[Interpolation Step for Error-Erasure Decoding]\label{prob:interpolation_erasures}
Let $\qtrafo{y}^{(i)}(x)$, $\forall i \in \intervallincl{1}{s}$, as in \eqref{eq:errorerasure_modifiedreceivedword}, and $g_0,g_1,\dots,g_{n-1} \in \Fqm$, which are linearly independent over $\Fq$, be given. 
Find an $(s+1)$-variate linearized polynomial of the form 
\begin{equation*}
Q(x,y_1,\dots,y_s) = Q_0(x) + Q_1(y_1) + \dots + Q_s(y_s), 
\end{equation*}
which satisfies for given integers $n$, $\tau,k^{(1)}, \dots, k^{(s)}$, $\numbRowErasures^{(1)}, \dots, \numbRowErasures^{(s)} $, $\numbColErasures$:
\begin{itemize}
\item[$\bullet$] $Q(g_j,\qtrafo{y}^{(1)}(g_j),\qtrafo{y}^{(2)}(g_j),\dots,\qtrafo{y}^{(s)}(g_j)) = 0$, $\quad\forall j \in \intervallexcl{0}{n}$,
\item[$\bullet$] $\deg_q Q_0(x) < n-\tau$,
\item[$\bullet$] $\deg_q Q_i(y_i)< n-\tau-(k^{(i)}-\numbColErasures-\numbRowErasures^{(i)}-1)$, $\quad \forall i \in \intervallincl{1}{s}$.
\end{itemize}
\end{problem_bold}

Similar to Lemma~\ref{lem:nonzero_interpolpoly}, a non-zero interpolation polynomial $Q(x,y_1,\dots,y_s)$, 
which satisfies the above mentioned conditions, exists if
 \begin{equation*} 
\tau < \frac{sn-\sum_{i=1}^{s}(k^{(i)}+\numbRowErasures^{(i)}+\numbColErasures)+s}{s+1}.
 \end{equation*}
If $k^{(i)}=k$, and $\numbRowErasures^{(i)} =\numbRowErasures$, $\forall i \in \intervallincl{1}{s}$, we obtain
 $\tau < s(n-k+1-\numbRowErasures-\numbColErasures)/(s+1)$.

The interpolation and root-finding procedure is straight forward to the errors-only approach from Section~\ref{subsec:intgab_interpolation_algo}. 
The error-erasure list decoder therefore returns all $\Lambda^{(i,R)} \big(f^{(i)}(\qreciproc{\Gamma^{(C)}}(x^{[\numbColErasures]}))\big)$, $\forall i \in \intervallincl{1}{s}$, such that $(f^{(1)}(\vec{g})^T \ f^{(2)}(\vec{g})^T\ \dots \ f^{(s)}(\vec{g})^T)^T$ is in rank distance at most $\tau$. 
In order to obtain $f^{(i)}(x)$, we have to divide from the left and right by $\Lambda^{(i,R)}(x)$ and $\qreciproc{\Gamma^{(C)}}(x^{[\numbColErasures]})$, respectively, $\forall i\in \intervallincl{1}{s}$.

The unique decoder can be modified in a similar way and returns with high probability the unique solution if 
 \begin{equation*} 
\tau \leq \Bigg\lfloor\frac{sn-\sum_{i=1}^{s}(k^{(i)}+\numbRowErasures^{(i)}+\numbColErasures)}{s+1}\Bigg\rfloor.
 \end{equation*}
With this principle, our interpolation-based decoding algorithm can be applied to (unique or list) error-erasure decoding of interleaved Gabidulin codes.

For interpolation-based error-erasure decoding in Hamming metric it is more common to puncture the code at the erased positions and interpolate an (interleaved) code of smaller length and same dimension(s) as the original code, whereas we interpolate a code of same length as the original code, but higher dimension(s).

\section{Conclusion and Outlook}\label{sec:conclusion}
This paper considers decoding approaches for interleaved Gabidulin codes. 
First, we have described two known decoding principles \cite{Loidreau_Overbeck_Interleaved_2006,SidBoss_InterlGabCodes_ISIT2010} and have proven a relation between them.
Second, we have shown a new approach for decoding interleaved Gabidulin codes based on interpolating a multi-variate linearized polynomial.
The procedure consists of two steps: an interpolation step and a root-finding step, where 
both can be accomplished by solving a linear system of equations. 
Our decoder can be used as a list decoder or as a probabilistic unique decoder. To our knowledge, it is the first list decoding algorithm for interleaved Gabidulin codes.
The complexity of the unique decoder as well as finding a basis of all solutions of the list decoder is quadratic in the length of the code; however, for the list decoder, finding the explicit list might require exponential time complexity.
The output of both decoders is a unique decoding result with high probability. 
Further, we have derived a connection to the two known approaches for decoding interleaved Gabidulin codes. 
This relation provides an upper bound on the failure probability of our unique decoder.
Finally, we have generalized our decoding principle such that it incorporates also row and column erasures.

For future work, it should be possible to apply re-encoding in order to reduce the complexity and to use subspace evasive subsets for the elimination of the valid solutions of the list decoder, similar to \cite{GuruswamiWang-LinearAlgebraicForVariantsofReedSolomonCodes_2012}.


%
%
%
%
%
%
%
%
%

\section*{Acknowledgment}
The authors thank Vladimir Sidorenko for the valuable discussions and the reviewers for their suggestions that helped to improve the presentation of the paper.

\section*{Appendix}
\begin{lemma}[Row Space of Composition]\label{lem:eval_row_space}
Let ${a}(x)$ and ${b}(x)$ denote two linearized polynomials in $\Linpolyring$ with $\deg_q {a}(x),\deg_q {b}(x) < m$. Let $c(x) = b(a(x))$ and let $\NormbasisOrdered = (\beta_0 \ \beta_1 \ \dots \ \beta_{m-1})$ be a basis of $\Fqm$ over $\Fq$.
Let $\A \in \Fq^{m \times m}$, $\Mat{C}\in \Fq^{m \times m}$ denote the matrix representations according to $\Basis$ of
\begin{align*}
\vecevalm{a}{\beta},\quad
\vecevalm{c}{\beta},
\end{align*}
respectively. Then, for the row spaces the following holds:
\begin{equation*}
\Rowspace{\Mat{C}} \subseteq \Rowspace{\Mat{A}}.
\end{equation*}
\end{lemma}

\begin{proof}
Consider the linearized polynomials as linear maps over $\Fqm$. Then, the kernel of the map ${a}$ is equivalent to the set of roots of ${a}(x)$ in $\Fqm$, considered as a vector space over $\Fq$. Since the roots of ${a}(x)$ are also roots of $c(x) = {b}({a}(x))$, the kernels are connected by $\ker({a}) \subseteq \ker(c)$. 
For the right kernels of the matrices $\ker(\mathbf A) \subseteq \ker(\C)$ holds, and the row spaces are related by $\Rowspace{\C} \subseteq \Rowspace{\mathbf A}$.
\qed\end{proof}

\end{document}